\documentclass[5p,times]{elsarticle}
\usepackage{hyperref}

\journal{Applied Intelligence}  

\usepackage{float}

\usepackage{balance}
\usepackage{listings}
\usepackage{color}

\definecolor{dkgreen}{rgb}{0,0.6,0}
\definecolor{gray}{rgb}{0.5,0.5,0.5}
\definecolor{mauve}{rgb}{0.58,0,0.82}

\makeatletter
\def\lst@makecaption{%
  \def\@captype{table}%
  \@makecaption
}
\makeatother

\graphicspath{{img/}}

\usepackage{verbatim}
\usepackage{url}
\usepackage{multirow}
\usepackage{algpseudocode}
\usepackage{graphicx}
\usepackage{epstopdf}
\usepackage{amssymb}
\usepackage{amsmath}
\newtheorem{definition}{Definition}
\newtheorem{proof}{Proof}    
\newtheorem{theorem}{Theorem}   
\usepackage{algorithm}
\usepackage{algorithmicx}
\usepackage{array}



\usepackage[caption=false,font=footnotesize,labelfont=sf,textfont=sf]{subfig}

\begin{document}
\begin{frontmatter}
	
	\title{FRI-Miner: Fuzzy Rare Itemset Mining}
	
	\author{Yanling Cui$ ^{1} $, Wensheng Gan$ ^{2,3}$*,  Hong Lin$ ^{4} $,  and  Weimin Zheng$ ^{1} $*}
	
	\address{$ ^{1} $College of Computer Science and Engineering, Shandong University of Science and Technology, Qingdao 266590, Shandong, China} 
	\address{$ ^{2} $College of Cyber Security, Jinan University, Guangzhou 510632, Guangdong, China}
	\address{$ ^{3} $Guangdong Artificial Intelligence and Digital Economy Laboratory (Guangzhou), Guangdong, China}

	\address{$ ^{4} $Department of Computer Sciences, Guangdong University of Technology, Guangdong 510006, China}

	\address{Email: ylcui001@gmail.com, wsgan001@gmail.com, lhed9eh0g@gmail.com, zhengweimin@sdust.edu.cn} 
	
	\cortext[cor1]{Corresponding author. Email: wsgan001@gmail.com} 

\begin{abstract}

Data mining is a widely used technology for various real-life applications of data analytics and is important to discover valuable association rules in transaction databases. Interesting itemset mining plays an important role in many real-life applications, such as market, e-commerce, finance, and medical treatment. To date, various data mining algorithms based on frequent patterns have been widely studied, but there are a few algorithms that focus on mining infrequent or rare patterns. In some cases, infrequent or rare itemsets and rare association rules also play an important role in real-life applications. In this paper, we introduce a novel fuzzy-based rare itemset mining algorithm called FRI-Miner, which discovers valuable and interesting fuzzy rare itemsets in a quantitative database by applying fuzzy theory with linguistic meaning. Additionally, FRI-Miner utilizes the fuzzy-list structure to store important information and applies several pruning strategies to reduce the search space. The experimental results show that the proposed FRI-Miner algorithm can discover fewer and more interesting itemsets by considering the quantitative value in reality. Moreover, it significantly outperforms state-of-the-art algorithms in terms of effectiveness (w.r.t. different types of derived patterns) and efficiency (w.r.t. running time and memory usage).

\end{abstract}
	
\begin{keyword}
	Quantitative data, fuzzy-set theory, rare pattern, fuzzy data mining
\end{keyword}
	
\end{frontmatter}


\section{Introduction}

Association mining \cite{agrawal1994fast,gan2017data,han2011data} of a transaction database is performed to determine association rules between a set of itemsets, for example, a set of events containing \{\textit{diaper}, \textit{beer}\} and \{\textit{milk}, \textit{bread}\}, which are often analyzed in market basket analysis. Several useful and interesting phenomena can be explored based on the association rules. By discovering the connection between items/objects, we can create a suitable market plan, form a suitable marketing strategy \cite{berry2004data, linoff2011data, shaw2001knowledge}, and effectively apply the data to all types of analysis. Different data-mining methods are used depending on the requirements of various applications. The examples of such methods include frequent pattern mining (FPM) \cite{agrawal1994fast, han2011data,lin2015rwfim}, utility-driven pattern mining \cite{fournier2019mining,gan2019survey,gan2018survey, nguyen2019efficient}, sequential pattern mining \cite{kim2007squire,le2018mining, srikant1996mining}, and rare pattern mining (RPM) \cite{kiran2011novel,koh2016unsupervised,szathmary2007towards,tsang2011rp}. FPM is commonly adopted to extract association rules from a transaction database. In general, association rules are categorized as ``frequent'' or ``rare'' according to the specified minimum support threshold (\textit{minSup}).  ``Frequent'' refers to common or anticipated phenomena, while ``rare'' represents infrequent or previously unknown phenomena; thus, varied information can be extracted from the database.

In real life, it is possible to buy multiple copies of the same item in a transaction database, and mining association rules from such a quantitative database is an important task. The fuzzy set theory, which was first proposed by Zadeh in 1965 \cite{zadeh1965fuzzy}, is more suitable for dealing with quantitative values and expressing appropriate language values because it can help people better understand knowledge. Each element in the fuzzy set can be assigned a membership degree to indicate the degree to which it belongs, such as \{\textit{A.low}\}, \{\textit{B.mid}\}.  Chan and Au \cite{chan1997mining} first proposed an Apriori-like \cite{agrawal1994fast} algorithm, namely F-APACS, to discover fuzzy association rules. Kuok et al. \cite{kuok1998mining} proposed a method for processing quantitative data. Hong et al. \cite{hong1999data} proposed an FDTA algorithm to process a quantitative database. In contrast to these Apriori-like algorithms, several fuzzy pattern mining algorithms based on tree structures were proposed. Lin et al. \cite{lin2010linguistic} proposed a fuzzy frequent pattern tree that can effectively discover fuzzy frequent itemsets (FFIs). Lin et al. \cite{lin2010efficient} proposed a compressed FFP (CFFP)-tree algorithm. Although this CFFP-tree-based algorithm can reduce the number of tree nodes by using an additional array on each node, it is requires computational expenses to save the array. In addition, if the transactions are large, the spatial complexity of each node is high. Therefore, Lin et al. \cite{lin2014mining} proposed the UBFFP-tree algorithm to solve the problem of CFFP-tree overhead, which uses the same global sort strategy as in CFFP-tree to construct trees. Each term in the transaction is obfuscated by retaining only the language terms with the largest cardinality later in the process. Subsequently, Lin et al. \cite{lin2015fast} proposed the FFI-Miner algorithm to discover a complete set of FFIs without generating candidates. By adopting the fuzzy-list structure, the necessary information is preserved in the mining process. They also proposed an effective pruning strategy that can reduce the search space and further accelerate the mining process. Some effective algorithms for mining FFIs for association rules are still being studied.

Fuzzy-based data mining is simple and similar to human reasoning. According to the algorithms mentioned above, a fuzzy theory has been extensively developed \cite{zadeh1965fuzzy} for application to FPM. However, the rare pattern mining (RPM)  remains to be explored further. FPM, which is relatively mature, can find often-appearing or  expected phenomena. In contrast, RPM can usually discover unknown or unexpected phenomena. Based on the existing algorithms, it is known that RPM can discover phenomena that are important in real life. For example, we assume that \textit{A} and \textit{B} represent two symptoms, and \textit{C} represents a disease. In general, symptom \textit{A} may lead to disease \textit{C}. After rare pattern mining, we may find that symptom \textit{B} also leads to disease \textit{C}, which will be of great help to the medical industry. The performance of students with poor scores can also be determined, and their learning conditions can be adjusted accordingly. Apriori-like algorithms were initially used to determine frequent or rare association rules. If the minimum support (\textit{minSup}) threshold is set extremely small, explosive growth occurs. However, if it is set extremely large, some useful association rules may be ignored. Based on the above problems, we attempt to use two thresholds, the minimum rare support (\textit{minRSup}) and the minimum frequent support (\textit{minFSup}), to achieve a better effect. As mentioned before, a quantitative transaction database can be processed with linguistic meaning using fuzzy theory. In the FFI-Miner algorithm \cite{lin2015fast}, which aims to mine FFIs, while fuzzy rare but quite interesting itemsets are ignored. 

To the best of our knowledge, RPM based on fuzzy theory has not yet been studied.  To this end, a novel algorithm named FRI-Miner for mining fuzzy rare itemsets (FRIs) from a quantitative transaction database is proposed in this paper.  The major contributions of this study are as follows.

\begin{enumerate}
	
	\item  This study is the first to formulate the problem of fuzzy RPM with linguistic meaning. The proposed fuzzy-list-based FRI-Miner algorithm addresses this problem successfully. As a fuzzy-theoretic data-driven model, it is explainable and similar to human reasoning that is more useful for decision making.
	
	\item Fuzzy rare itemsets in FRI-Miner are categorized into three types: (1) containing only fuzzy rare items; (2) containing any combination of fuzzy rare items and fuzzy frequent items; and (3) containing only fuzzy frequent items. The first type is easy to understand; in the third type, fuzzy frequent itemsets that are themselves frequent may actually be fuzzy rare itemsets.
	
	\item Several pruning strategies that utilize the properties of fuzziness, rare pattern, and fuzzy support are designed to successfully reduce the search space of FRI-Miner.
	
	\item Experiments on several benchmark databases are conducted to show that the proposed FRI-Miner algorithm has better effectiveness and mining efficiency compared to those of existing methods.
	
\end{enumerate}

The remainder of this paper is organized as follows: In Section \ref{sec:relatedwork}, we review previous related work in the field of FPM and RPM. In Section \ref{sec:preliminaries}, we present the definition and basic concepts of rare itemsets and fuzzy theory and then formulate the problem of FRI mining. The details of the proposed FRI-Miner algorithm are presented in Section \ref{sec:algorithm}. In addition, we use an example to illustrate the details of FRI-Miner for ease of understanding. The experimental evaluation is provided in detail in Section \ref{sec:experiments}. Finally, the conclusions and future work are presented in Section \ref{sec:conclusion}.

\section{Literature Review}
\label{sec:relatedwork}

In this section, we briefly discuss fuzzy pattern mining and RPM. Then, we highlight the importance of fuzzy RPM and further discuss several applications.

\subsection{Fuzzy pattern mining}

According to the existing FPM algorithms, the data types can be roughly categorized as Boolean, ordered, or quantitative \cite{gan2017data,han2011data,lin2015rwfim}. Traditional FPM aims at processing data that do not contain quantitative values, although quantitative data are common in real life.  To solve the processing problem of a quantitative database, the quantitative value of a project can be converted into a language term with a fuzzy degree by using fuzzy set theory \cite{zadeh1965fuzzy}. Meaningful fuzzy association rules can be discovered in real life. The addition of fuzzy concepts fuzzifies the quantitative values. This is also convenient to better understand the true meaning of the data and make the data simpler and easier to understand. For example, if the database does not contain the quantity value corresponding to each item, then the frequency of the item is only related to the Boolean value. However, for databases containing quantitative values, the frequency of occurrence of each item is determined by the quantitative values of the items.

Frequent itemset mining based on fuzzy theory was further elaborated by the Apriori algorithm \cite{agrawal1994fast}. Fuzzy association rules are induced by the hierarchical intelligent mining of the FFI. First, according to the predefined membership function, the quantization value of the item is converted into a language term, which not only reflects the occurrence frequency of the item in the database, but also reflects the support of the itemset. Therefore, several algorithms for mining fuzzy itemsets that meet the minimum support threshold through deterministic factors have been designed. The F-APACS algorithm for mining fuzzy association rules was first proposed by Chan and Au \cite{chan1997mining}. Kuok et al. \cite{kuok1998mining} proposed an algorithm by processing quantitative attributes to discover fuzzy association rules, and Hong et al. \cite{hong1999mining,hong1999data} proposed the FDTA algorithm to process quantitative databases using fuzzy set theory. Subsequently, a fuzzy frequent pattern tree \cite{papadimitriou2005fuzzy} was proposed. Several tree-based algorithms for mining FFIs have been proposed \cite{lin2014mining, lin2010linguistic}, which are based on the FP-growth method \cite{han2004mining} for mining frequent itemsets. These algorithms first construct a tree structure of fuzzy frequent items through the reserved frequent itemsets and then further extract more fuzzy frequent patterns from the constructed tree structure. However, the mining process with a pattern-growth mechanism \cite{han2004mining} is often complicated and requires the storage of numerous tree nodes with additional information. To summarize, all these fuzzy-based FPM algorithms can only discover expected frequent phenomena.

\subsection{Rare pattern mining}

Based on the Apriori algorithm \cite{agrawal1994fast}, many frequency-based data mining algorithms have been extended, but the mining results mostly correspond to common or expected phenomena. To address this, another data mining framework, rare pattern mining (RPM) \cite{ji2012method,koh2016unsupervised,sadhasivam2011mining}, has been introduced, and many algorithms for RPM have been proposed in recent years. In contrast to frequent and common patterns, discovering rare patterns may be more useful in some cases (e.g., itemsets and association rules), which are important for real-life applications.

Most of static rare itemset mining algorithms can be roughly divided into support threshold, no support threshold, and constraints \cite{koh2016unsupervised}. Because we use the support threshold in this study, we provide a brief overview of RPM with the support threshold. Because the support threshold for rare itemsets is lower than that for frequent itemsets, the generation of rare itemsets is better realized by setting lower or distinct support thresholds. In 1999, Liu et al. \cite{liu1999mining} proposed the multiple support a priori algorithm that adopted multiple minimum support thresholds to successfully discover rare itemsets. The Apriori-reverse algorithm proposed by Koh and Rountree \cite{koh2005finding} aims to discover rare rules with respect to completely dispersed itemsets, which only contain items below the maximum support threshold. Later, a rare itemset mining algorithm (ARIMA) was proposed by Szathmary et al. \cite{szathmary2007towards}. Troiano et al. \cite{troiano2009fast} introduced the rarity algorithm, which is a top-down rare itemset mining algorithm. Up to now, a number of RPM algorithms have been extensively proposed, such as CFP-growth++ \cite{kiran2011novel}. Among them, several studies are designed to deal with dynamic data streams \cite{hemalatha2015minimal, huang2012rare, huang2014detecting}. There are numerous candidates for most algorithms based on Apriori mechanism. The RP-tree-based algorithm \cite{tsang2011rp} discovers rare patterns that meet the conditions between \textit{minRSup} and \textit{minFSup}.

\subsection{Applications of fuzzy rare pattern mining}

For some real applications, unusual rare patterns are more important and useful than frequent ones. In some application domains,  RPM is more suitable for intelligent systems. The result of network intrusion can be obtained by detecting whether the network is abnormal. In medicine, sudden changes are diagnosed by finding data that are different from those corresponding to normal health. In an insurance company, by finding rare people who need high-risk claims, making reasonable marketing strategies, and so on. In recent years, the field of RPM \cite{koh2016unsupervised} has been further developed. The emergence of the RPM has also made a significant contribution to the research community. However, these previous RPM algorithms rarely involve fuzzy theory. In contrast, many FPM algorithms based on fuzzy theory have been widely used in the mining of quantitative data, while RPM dealing with quantitative data is extremely rare. To the best of our knowledge, RPM dealing with Boolean or quantitative data based on fuzzy theory has not yet been studied. To this end, we propose an effective algorithm that uses fuzzy theory to discover the interesting rare patterns from a quantitative transaction database.

\section{Preliminaries and Problem Formulation}
\label{sec:preliminaries}

In this section, we introduce some basic concepts, principles of fuzzy-driven pattern mining and RPM. Some definitions in previous research are adopted here to present common concepts clearly. Further details regarding the background of the fuzzy FPM can be found in Ref. \cite{lin2010linguistic,lin2015fast}.

\subsection{Preliminaries}

\textit{I} = \{$ i_{1} $, $ i_{2} $, $\dots$, $ i_{m} $\} represents a finite set $I$ composed of $m$ different items. \textit{D} = \{$ t_{1} $, $ t_{2} $, $\dots$, $ t_{n} $\} represents a transactional database \textit{D} ($1 \leq q \leq n$) composed of \textit{n} different items, in which each transaction $ t_{q} $ is a subset of \textit{I}. Each transaction contains a unique identifier \textit{tid}, and each item lists the number of quantities as $ v_{iq} $. If an itemset consists of \textit{k} different items, we call it \textit{k}-itemset, and if each item in the \textit{k}-itemsets is included in the transaction $ t_{q} $, we call the itemset a subset of the transaction $ t_{q} $. The minimum support is defined as \textit{minRSup}, and the maximum support is defined as \textit{minFSup}. The specified member function is set to $\mu$, which can be adjusted according to the user's needs. In this study, all items in our database are represented using a single fuzzy membership function, as shown in Figure \ref{fig:3terms}. We set three language terms in the adopted membership functions: \textit{low}, \textit{middle}, and \textit{high}. The quantitative value of each item is obscured by the same member function as the corresponding language term.

\begin{figure}[!htbp]
	\centering
	\includegraphics[scale=0.45]{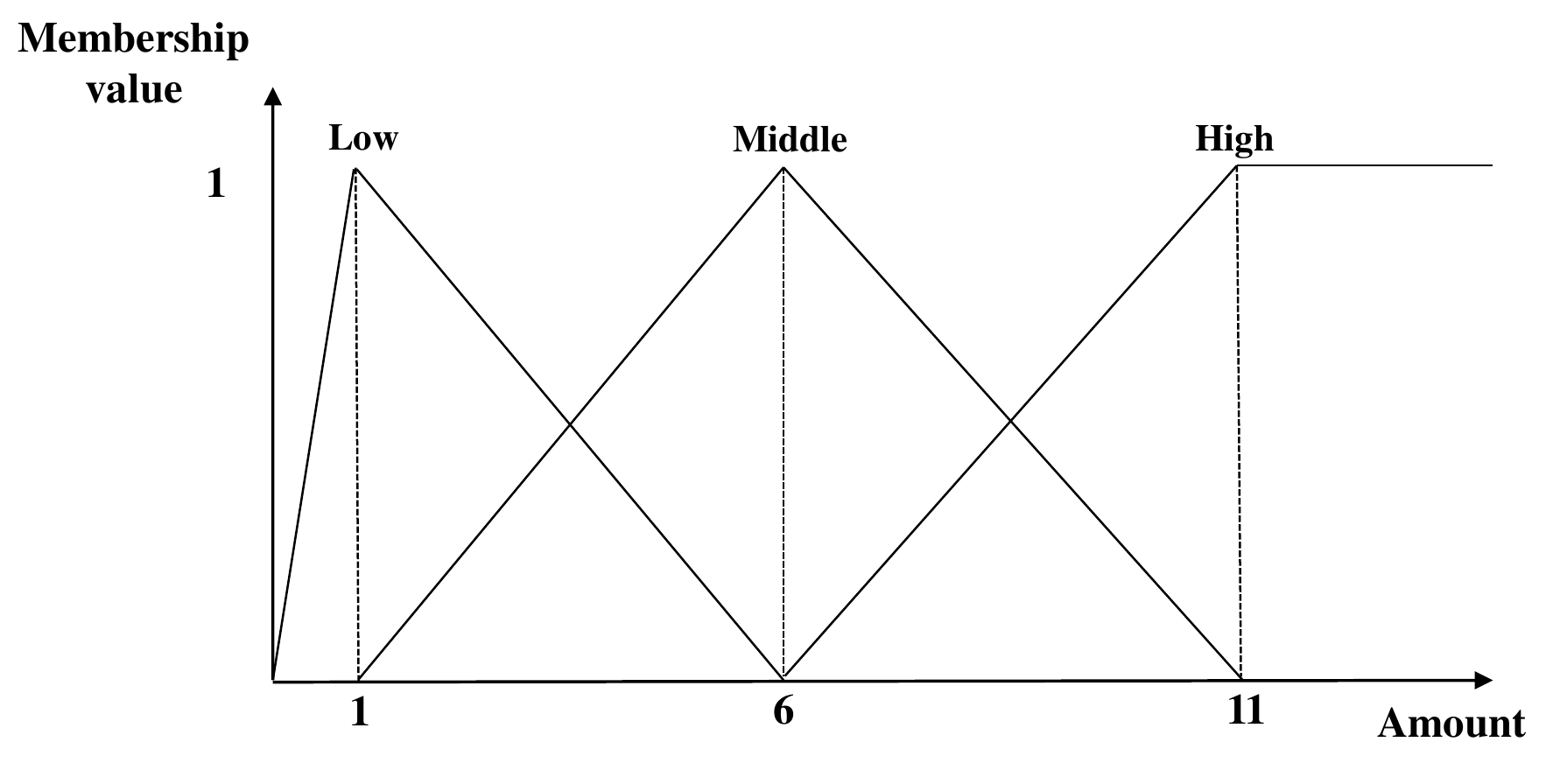}
	\captionsetup{justification=centering}
	\caption{The used linear membership functions of linguistic 3-terms.}
	\label{fig:3terms}
\end{figure}

In this paper, we present a running example using a quantitative transactional database, as shown in Table \ref{table:db}. There are six items, (\textit{A}), (\textit{B}), (\textit{C}), (\textit{D}), (\textit{E}), (\textit{F}), and eight transactions \{$ t_{1} $, $ t_{2} $, $\dots$, $ t_{8} $\}. We set the minimum support threshold to \textit{minRSup} (= 25\%) and the maximum support threshold to \textit{minFSup} (= 50\%). The idea is to filter out items that seem insensible and similar to noise; thus, this can help to reduce unnecessary search space.

\begin{table}[h]
	\begin{center}
		\caption{A quantitative database.}
		\label{table:db}
		\begin{tabular}{|c|c|}
			\hline
			\textbf{\textit{tid}} & \textbf{Transaction}  \\ \hline
			$t_{1}$	& $A$:3, $B$:5, $D$:10, $E$:9  \\ \hline
			$t_{2}$	& $B$:8, $D$:3  \\  \hline
			$t_{3}$	& $A$:3, $B$:8, $D$:9, $F$:5  \\  \hline
			$t_{4}$	& $B$:5, $C$:4, $D$:11, $E$:2 \\  \hline
			$t_{5}$	& $B$:7, $C$:3, $D$:5, $F$:3  \\  \hline
			$t_{6}$	& $A$:2, $B$:5, $C$:3, $D$:7  \\  \hline
			$t_{7}$	& $A$:2, $B$:4, $D$:9, $F$:2  \\  \hline
			$t_{8}$	& $B$:5, $C$:2, $D$:10, $E$:3 \\  \hline
		\end{tabular}
	\end{center}
\end{table}

\begin{definition}
	\textbf{(The attributes of a quantitative database)} 
	\rm  They are represented by the language variable  $ R_{i} $, and the fuzzy language terms are represented by the natural language  as ($ R_{i1} $, $ R_{i2} $, $\dots$, $ R_{il} $). $ R_{i} $ can be defined by the membership function $\mu$.
\end{definition}

For example, in this study, the membership function  $\mu$ applied in the running example is shown in Figure \ref{fig:3terms}. There are five items in Table \ref{table:db}: (\textit{A}), (\textit{B}), (\textit{C}), (\textit{D}), (\textit{E}), and (\textit{F}). The three language terms are expressed as $ Low (L) $, $ Middle (M) $, and $ High (H) $. In transaction $t_1$, item \textit{B} is denoted by $ (B.L) $, $ (B.M) $, and $ (B.H) $, as  shown in Figure \ref{fig:3terms}. The other items in these transactions are calculated similarly to those of item \textit{B}.

\begin{definition}
\textbf{(The quantitative value of an item)} 
\rm  In the quantitative database, the value of an item is expressed as $ v_{iq}$, which represents the number of this item \textit{i} in transaction $t_{q}$.
\end{definition}

For example, in transaction $ t_{2} $, the quantitative values of the items as (\textit{B}) and (\textit{D}) are $ v_{B2}$ (= 8) and $ v_{D2} $ (= 3), respectively. Here, we can clearly see what it means.

\begin{definition}
\textbf{(Fuzzy set)} 
\rm  In the fuzzy stage, a fuzzy set refers to the set of fuzzy language terms with a membership degree (fuzzy value) converted from the quantitative value $ v_{iq} $ of the item \textit{i} in the transaction database $ t_{q} $ by the membership function $\mu$. Its specific expression is as follows:	
\begin{equation}
	f_{iq} = \mu_{i}(v_{iq}) = (\frac{fv_{iq1}}{R_{i1}} + \frac{fv_{iq2}}{R_{i2}} + \dots + \frac{fv_{iqh}}{R_{il}}),
\end{equation}
where \textit{l} refers to the number of fuzzy language terms in which the membership function $\mu$ converts to \textit{I}, $ R_{il} $ represents the \textit{l}-th fuzzy language term of item \textit{i}, and $ fv_{iql} $ represents the membership of item \textit{i} in the quantitative value $ v_{iq} $ (fuzzy value) of the \textit{l}-th fuzzy language term $ R_{il}$'s, where $ fv_{iql}\subseteq [0, 1]$.

\end{definition}

For example, in transaction $t_2$, we convert the quantitative value of the items in Table \ref{table:db} to the membership degree using the membership function. Here, we provide a detailed description of the quantitative value of (\textit{B}) in Figure \ref{fig:3terms} as ($\frac{0.6}{B.L} + \frac{0.4}{B.M}$). The transformation of all other transactions is similar to the transformation of item (\textit{B}) in transaction $t_2$. The specific results are listed in Table \ref{table:tansformed}.

\begin{table}[htb]
	\begin{center}
		\caption{Transformed results from Table \ref{table:db}.}
		\label{table:tansformed}		
		\begin{tabular}{|c|c|} 
			\hline
			\textbf{\textit{tid}} & \textbf{Fuzzy transaction}  \\ \hline
			$t_{1}$	&  $\frac{0.6}{A.L}+\frac{0.4}{A.M}$, $\frac{0.2}{B.L}+\frac{0.8}{B.M}$, $\frac{0.2}{D.M}+\frac{0.8}{D.H}$, $\frac{0.4}{E.M}+\frac{0.6}{E.H}$\\ \hline
			$t_{2}$	&  $\frac{0.6}{B.M}+\frac{0.4}{B.H}$, $\frac{0.6}{D.L}+\frac{0.4}{D.M}$\\ \hline
			$t_{3}$	&  $\frac{0.6}{A.L}+\frac{0.4}{A.M}$, $\frac{0.6}{B.M}+\frac{0.4}{B.H}$, $\frac{0.4}{D.M}+\frac{0.6}{D.H}$, $\frac{0.2}{F.L}+\frac{0.8}{F.M}$\\ \hline
			$t_{4}$	&  $\frac{0.2}{B.L}+\frac{0.8}{B.M}$, $\frac{0.4}{C.L}+\frac{0.6}{C.M}$, $\frac{0}{D.M}+\frac{1}{D.H}$, $\frac{0.8}{E.L}+\frac{0.2}{E.M}$\\ \hline
			$t_{5}$	&  $\frac{0.8}{B.M}+\frac{0.2}{B.H}$, $\frac{0.6}{C.L}+\frac{0.4}{C.M}$, $\frac{0.2}{D.L}+\frac{0.8}{D.M}$, $\frac{0.6}{F.L}+\frac{0.4}{F.M}$\\ \hline
			$t_{6}$	&  $\frac{0.8}{A.L}+\frac{0.2}{A.M}$, $\frac{0.2}{B.L}+\frac{0.8}{B.M}$, $\frac{0.6}{C.L}+\frac{0.4}{C.M}$, $\frac{0.8}{D.M}+\frac{0.2}{D.H}$\\ \hline
			$t_{7}$	& $\frac{0.8}{A.L}+\frac{0.2}{A.M}$, $\frac{0.4}{B.L}+\frac{0.6}{B.M}$, $\frac{0.4}{D.M}+\frac{0.6}{D.H}$, $\frac{0.8}{F.L}+\frac{0.2}{F.M}$\\ \hline
			$t_{8}$	&  $\frac{0.2}{B.L}+\frac{0.8}{B.M}$, $\frac{0.8}{C.L}+\frac{0.2}{C.M}$, $\frac{0.2}{D.M}+\frac{0.8}{D.H}$, $\frac{0.6}{E.L}+\frac{0.4}{E.M}$\\ \hline
		\end{tabular}
	\end{center}
\end{table}

\begin{definition}
\textbf{(Membership degree)} 
\rm  The support after the membership function is converted to membership degree and expressed as $ sup(R_{il}) $. The sum of the scalar cardinality of the fuzzy value of $ R_{il} $ is expressed as:	
\begin{equation}
	sup(R_{il})=\sum_{R_{il}\subseteq t_{q}\wedge t_{q}\in Q'}fv_{iql}.
\end{equation}
where database $ Q' $ is the database transformed by the  membership function $\mu$, which is the same as the original database $D$.

\end{definition}

For example, the support of the fuzzy terms ($ B.M $) appears in transactions $ t_{1} $, $ t_{2} $, $ t_{3} $, $ t_{4} $, $ t_{5} $, $ t_{6} $, $ t_{7} $, and $ t_{8} $, as shown in Table \ref{table:tansformed}. Thus, its support in the running database can be calculated as $ sup(B.M) $ = 0.8 + 0.6 + 0.6 + 0.8 + 0.8 + 0.8 + 0.6 + 0.8 = 5.8.

\begin{definition}

\textbf{(The minimum fuzzy value of an \textit{k}-itemset)}
\rm  An itemset \textit{X} is composed of \textit{k}-items ($k$  $\geq $ 1), and its support degree is expressed as $ sup(X) $, which represents the sum of the minimum fuzzy values of \textit{k}-items contained in \textit{X}. The definition is as follows:
\begin{equation}
	\begin{aligned}
	sup(X) & = \{X\in R_{il}| \sum_{X\subseteq t_{q}\wedge t_{q}\in Q'}min(fv_{i_{j}ql}, fv_{i_{m}ql}), i_{j}, i_{m}\in X, i_{j}\notin i_{m}\},
	\end{aligned}
\end{equation}
where the items in \textit{X} do not intersect each other.
\end{definition}

For example, the fuzzy 2-itemset $ (A.L, D.H) $ appears in transactions $ t_{1} $, $ t_{3} $, $ t_{6} $, and $ t_{7} $, as shown in Table \ref{table:tansformed}. Thus, the support of $ (A.L, D.H) $ can be calculated as $ sup(A.L, D.H) $ = \{\textit{min}(0.6, 0.8) + \textit{min}(0.6, 0.6) + \textit{min}(0.8, 0.2) + \textit{min}(0.8, 0.6)\} = \{0.6 + 0.6 + 0.2 + 0.6\} = 2.0.

\begin{definition}
\textbf{(Rare itemset)} 
\rm In a transaction database, the mining of some common itemsets aims at discovering frequent itemsets. In previous methods, a minimum frequent support \textit{minFSup} is generally set, and itemsets that meet the \textit{minFSup} specified by the user are preserved. Thus, more frequent itemsets that meet these criteria are discovered. For RPM, we set the minimum rare support \textit{minRSup}. RPM discovers itemsets that meet the criteria of \textit{minRSup}. However, the support of these rare itemsets cannot be greater than that of \textit{minFSup}.
\end{definition}

For example, we assume that \textit{sup($A$)} = 1.0, and \textit{sup($B$)} = 2.5. If we set the \textit{minRSup} to 2.0, and the \textit{minFSup} to 3.0, then, we find that item $A$ does not meet the \textit{minRSup}. The support of item $B$ is greater than that of \textit{minRSup} and less than \textit{minFSup}. Thus, item $B$ is a rare pattern that we intend to discover.

\begin{definition}
\textbf{(Fuzzy rare itemset)} 
	\rm  Considering quantitative data, the user specifies the specific minimum rare support threshold and minimum frequent support threshold, which are \textit{minRSup} and \textit{minFSup}, respectively. Only the itemsets that satisfy two conditions can be considered as fuzzy rare itemsets (FRIs), as follows:	
	\begin{equation}
		FRIs \leftarrow \{X|\ minRSup \times |D|\ \leq sup(X) \leq minFSup \times |D|\}.
	\end{equation}
\end{definition}

We take an example of 1-itemset. Through the calculation of the fuzzy values, we find that  $ (A.L) $ = 2.8, $ (B.M) $ = 5.8, $ (C.L) $ = 2.4, and $ (D.H) $ = 4.0. According to the definition of FRIs, we find that $ (A.L) $ and $ (C.L) $ are rare itemsets, whereas $ (B.M) $ and $ (D.H) $ are frequent itemsets.

\subsection{Problem statement}

To date, many algorithms have been developed for mining frequent patterns in quantitative databases. Corresponding with the emergence of quantitative values, fuzzy theory has also merged. Therefore, a mining algorithm for fuzzy frequent patterns  is proposed accordingly. However, there are many interesting but rare patterns that are ignored, mostly because of setting a low support. 

In this study, the problem of fuzzy rare pattern mining (fuzzy RPM for short) is formulated as follows. Given a quantitative database, the specific minimum support threshold and maximum support threshold, which are \textit{minRSup} and \textit{minFSup}, respectively. The goal of fuzzy RPM is to discover the complete set of FRIs that satisfies two conditions.

\section{Proposed Fuzzy Mining Algorithm: FRI-Miner}
\label{sec:algorithm}

For the mining of quantitative databases, previous studies have shown that fuzzy theory can play an important role. Through many previous statements, we find that it has been relatively mature in the mining of frequent itemsets, while RPM based on fuzzy theory has not yet been studied. In this study, we apply fuzzy theory to process the quantitative database, obtain the candidates, and finally discover FRIs by using the fuzzy-list structure \cite{lin2015fast}. In summary, the specific steps of the FRI-Miner algorithm include: 1) fuzzification phase, 2) construction of fuzzy-list, and 3) recursively mining FRIs, which are presented as follows.

\subsection{Fuzzification phase}

FRI-Miner first fuzzes the quantitative database and converts $ v_{iq} $ in quantitative data into the membership degree of the corresponding fuzzy terms through a membership function. Next, it forms a new fuzzy database by using the maximum scalar cardinality and support-ascending order.

\begin{definition}
	\textbf{(Maximum scalar cardinality)} 
	\rm  For an item \textit{i}, it uses the corresponding fuzzy term  $ R_{il} $ to denote its corresponding quantitative value $v_{iq}$. In the expressed fuzzy terms, we find the fuzzy terms that should be reserved according to the maximum scalar cardinality. This can represent the corresponding language variables as item \textit{i}.
 \end{definition}	

For example, in Table \ref{table:tansformed}, the transformed fuzzy terms with their summed fuzzy values of the linguistic variable (\textit{A}) are (\textit{A.L}: 2.8, \textit{A.M}: 1.2, \textit{A.H}: 0). Thus, the fuzzy term (\textit{A.L}) is used to represent the linguistic variable of (\textit{A}), which can be used for the later mining process of FRIs. 

\begin{table}[h]
	\begin{center}
		\caption{A revised database.}
		\label{table:reviseddb}		
		\begin{tabular}{|c|c|}
			\hline
			\textbf{\textit{tid}} & \textbf{Fuzzy transaction}  \\ \hline
			$t_{1}$	&   $\frac{0.6}{A.L}$, $\frac{0.8}{D.H}$, $\frac{0.8}{B.M}$ \\ \hline
			$t_{2}$ &  $\frac{0.6}{B.M}$\\ \hline
			$t_{3}$ &  $\frac{0.6}{A.L}$, $\frac{0.6}{D.H}$, $\frac{0.6}{B.M}$\\ \hline
			$t_{4}$	&   $\frac{0.4}{C.L}$, $\frac{1.0}{D.H}$, $\frac{0.8}{B.M}$\\ \hline
			$t_{5}$	&   $\frac{0.6}{C.L}$, $\frac{0.8}{B.M}$\\ \hline				
			$t_{6}$	&   $\frac{0.6}{C.L}$, $\frac{0.8}{A.L}$, $\frac{0.2}{D.H}$, $\frac{0.8}{B.M}$\\ \hline
			$t_{7}$	&   $\frac{0.8}{A.L}$, $\frac{0.6}{D.H}$, $\frac{0.6}{B.M}$\\ \hline
			$t_{8}$	&   $\frac{0.8}{C.L}$, $\frac{0.8}{D.H}$, $\frac{0.8}{B.M}$\\  \hline									
		\end{tabular}
	\end{center}
\end{table}

Based on the reserved $ R_{il} $ of the maximum scalar cardinality, the fuzzy quantitative database was modified to form a new fuzzy-based database. The revised database from Table \ref{table:tansformed} is presented in Table \ref{table:reviseddb}. Fuzzy items reserved by each transaction $ t_{q} $ are then sorted into ordered fuzzy items according to the support-ascending order. Note that the items retained here meet the minimum level of support. In addition to rare items, these items also have frequent items. The cross combination between them produces a new itemset that meets the two conditions.

\begin{definition}
	\textbf{(The support-ascending order)} 	
	\rm  In the fuzzy transaction database, fuzzy items that satisfy the fuzzy value condition are retained according to the maximum scalar cardinality in the transaction $ t_{q} $. Based on the fuzzy values of the reserved fuzzy items, they are sorted according to the support-ascending order to prepare for the subsequent computation.
\end{definition}

For example, we can obtain the result through the transformed fuzzy database, based on the maximum scalar cardinality and the support-ascending order.

\subsection{Fuzzy-list construction phase}

In FRI-Miner, first, $v_{iq}$ in a quantitative database is converted into a membership degree corresponding to the corresponding fuzzy terms through the membership function. Then, we form a novel fuzzy database and then construct the corresponding fuzzy-list structure \cite{lin2015fast}. We keep the fuzzy terms in $ L_{1} $, which uses three fields that make up the fuzzy-list: the transaction identifier $(tid) $, the internal fuzzy value $ (if) $, and the remaining fuzzy value $ (rf) $. We briefly introduce these details.

\begin{definition}
	\textbf{(Transaction identifier)} 
	\rm  It represents a fuzzy term $ R_{il} $ in transaction $ t_{q} $, which is a subset of the corresponding transaction in $ t_{q} $ is denoted as $ R_{il}\subseteq t_{q} $. Here, we use the corresponding $ (tid) $ to represent the presence in that transaction.
\end{definition}

According to the initial construction of the fuzzy-list, as shown earlier in $ L_{1} $, the support ascending order is used to obtain $(C.L < A.L < D.H < B.M) $ in Figure \ref{fig:liststructure}. For example, in Table \ref{table:reviseddb} and  Figure \ref{fig:liststructure}, the fuzzy item $ (B.M) $ exists in $ t_{1} $, $ t_{2} $, $ t_{3} $, $ t_{4} $, $ t_{5} $, $ t_{6} $, $ t_{7} $, and $ t_{8} $.

\begin{definition}
	\textbf{(Internal fuzzy value)} 
	\rm The fuzzy value $ (if) $ of the fuzzy term $ R_{il} $ in the transaction $ t_{q} $ is denoted as $ if(R_{il}, t_{q})$.	
\end{definition}

For example, in Table \ref{table:reviseddb}, the internal fuzzy values of $ (B.M) $ in $ t_{1} $ and $ t_{2} $ are $if(B.M, t_{1})$ = 0.8, and $if(B.M, t_{2})$ = 0.6, respectively.

\begin{definition}
	\textbf{(The resting fuzzy value  \cite{lin2015fast})} 
	\rm  In the fuzzy-list, the resting fuzzy value of $ R_{il} $ is expressed as $ rf(R_{il}, t_{q}) $. It represents the maximum fuzzy value obtained by performing a union operation. In other words, the upper bound value of all fuzzy terms $ R_{il} $ in $ t_{q} $, as shown in Figure \ref{fig:liststructure}. It is defined as:
	\begin{equation}
	rf(R_{il}, t_{q}) = max\{if(z, t_{q})  |z\in (t_{q}/R_{il})\}.
	\end{equation}
	
\end{definition}

According to Table \ref{table:reviseddb}, formed by the support-ascending order, the \textit{rf}($C.L$, $t_{4}$) is calculated as $ max\{0.8, 1.0\} $ = 1.0, and \textit{rf}($C.L$, $t_{5}$) is calculated as $ max\{0.8\} $ = 0.8, the \textit{rf}($C.L$, $t_{6}$) $ t_{6} $ is counted as $ max\{0.8,0.2,0.8\} $ = 0.8, and the \textit{rf}($C.L$, $t_{8}$) is $ max\{0.8,0.8\} $ = 0.8. Details of the fuzzy-list structure can be referred to Ref. \cite{lin2015fast}.

\begin{figure}[hbt]
	\centering
	\includegraphics[scale=0.75]{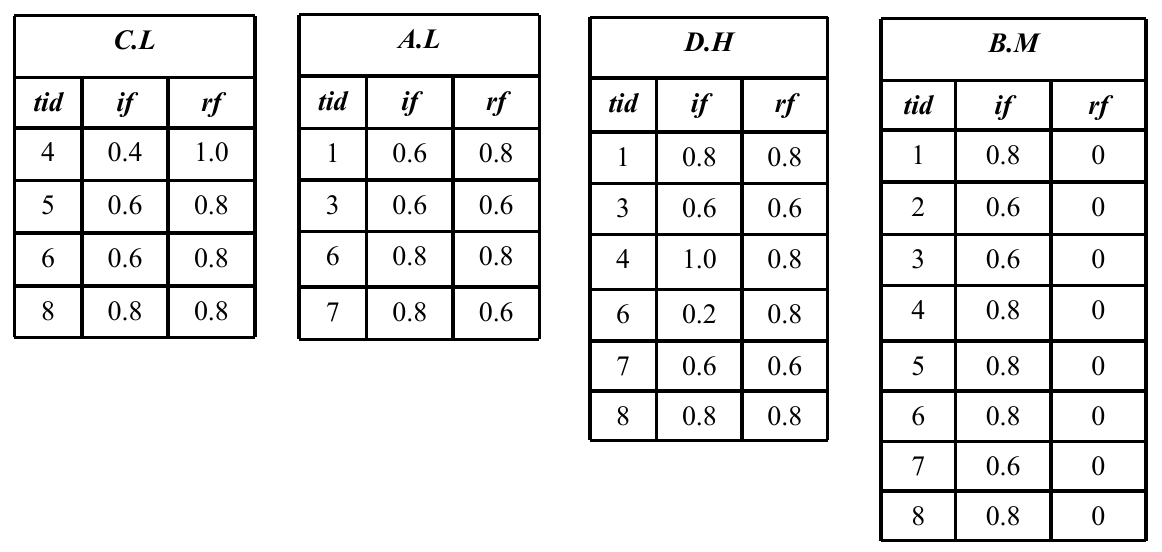}
	\captionsetup{justification=centering}
	\caption{The constructed fuzzy-list structures.}
	\label{fig:liststructure}
\end{figure}

In Figure \ref{fig:liststructure}, we can observe that the element (4, 0.4, 1.0) in the constructed fuzzy-list structure of $ (C.L) $ indicates $ (tid, if, rf) $, 4 represents the transaction $ t_{4} $, 0.4 represents the internal fuzzy value is 0.8, and 1.0 represents the resting fuzzy value after $ (C.L) $ is 1.0. The other information for the fuzzy item is similarly represented.

The fuzzy-list of the 1-itemset performs the intersection operation to form a new fuzzy-list structure of $ k $-itemsets ($ k \geq 2 $). In the recombination process, the ones with the same $ tid$ are combined. Except for the item that needs to be calculated, the internal fuzzy value of the transaction corresponding to all eligible fuzzy \textit{k}items ($ k\geq 2 $) takes the minimum value of the merged item as the remaining fuzzy value. In Figure \ref{fig:list}, the combined results of fuzzy 2-itemsets are shown as $(C.L,A.L)$, $(C.L,D.H)$, $(C.L,B.M)$, and so on.

\begin{figure*}[!hbtp]
	\centering
	\includegraphics[scale=0.5]{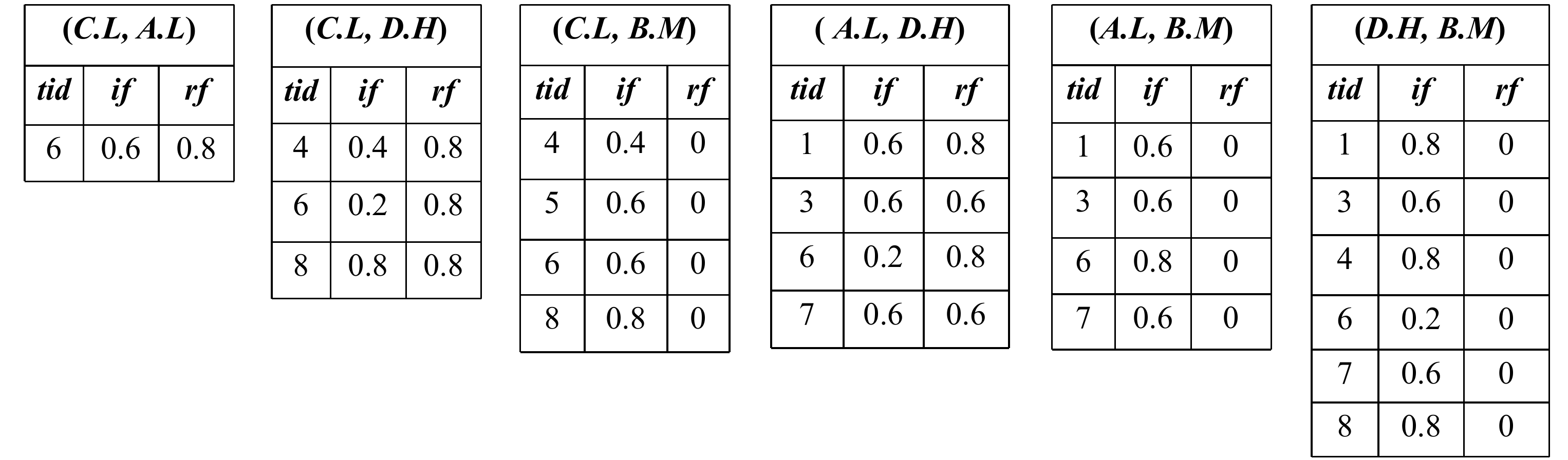}
	\captionsetup{justification=centering}
	\caption{The fuzzy-lists of fuzzy 2-itemsets.}
	\label{fig:list}
\end{figure*}

According to the similar fuzzy-list in Figure \ref{fig:liststructure}, we obtain the corresponding values of the three columns, and the corresponding support can be obtained according to the values of the second and third columns. They are defined as follows:

\begin{definition}
	\rm  In the fuzzy-list, we can calculate the sum of the inner fuzzy values of an itemset  $ R_{il} $ in a quantitative database, and it is denoted as  \textit{SUM}($R_{il}.if) $ \cite{lin2015fast} and is defined as follows:
	\begin{equation}
	SUM(R_{il}.if) = \sum_{R_{il} \subseteq t_{q} \wedge t_{q}\in Q'} if(R_{il}, t_{q}).
	\end{equation}
\end{definition}

For example, in Figure \ref{fig:liststructure}, the sum of the internal fuzzy values of $ (C.L) $ in $ D $ is calculated as (0.4 + 0.6 + 0.6 + 0.8) = 2.4.

\begin{definition}
\rm  In the fuzzy-list, we calculate the sum of the resting fuzzy values of $ R_{il} $ in the quantitative database $ D $ is defined as \textit{SUM(R$_{il}$.rf)}, \cite{lin2015fast}. Here, the $ rf $ value is calculated from the third column as follows:
\begin{equation}
	SUM(R_{il}.rf) = \sum_{R_{il} \subseteq t_{q} \wedge t_{q}\in Q'}rf(R_{il}, t_{q}).
\end{equation}
\end{definition}
	
For example, in Figure \ref{fig:liststructure}, the sum of the resting fuzzy values for $ (C.L) $ in $ D $ is calculated as (1.0 + 0.8 + 0.8 + 0.8) = 3.4.

The fuzzy value in the fuzzy-list was used to obtain the total support. Here, we use the support-ascending order to construct a set of fuzzy-lists. According to the measurement of the support of the corresponding item, we obtain the rare itemsets between the support tends to the minimum and the maximum specified support threshold, and the fuzzy itemsets with the support above the maximum support. We also obtain a set of items that we can intersect further in Figure \ref{fig:list}. Then, it sorts the results according to the support-ascending order and adopts the pruning strategies to prune the search space with respect to an  enumerated tree \cite{rymon1992search}. Then, the 1-itemset is gradually expanded to 2-itemsets, 3-itemsets, and other combinations, which are shown in Figure \ref{fig:enumeratetree}. Finally, a set of rare itemsets that satisfy support can be obtained.
	
\begin{figure}[!hbtp]
	\centering
	\includegraphics[scale=0.4]{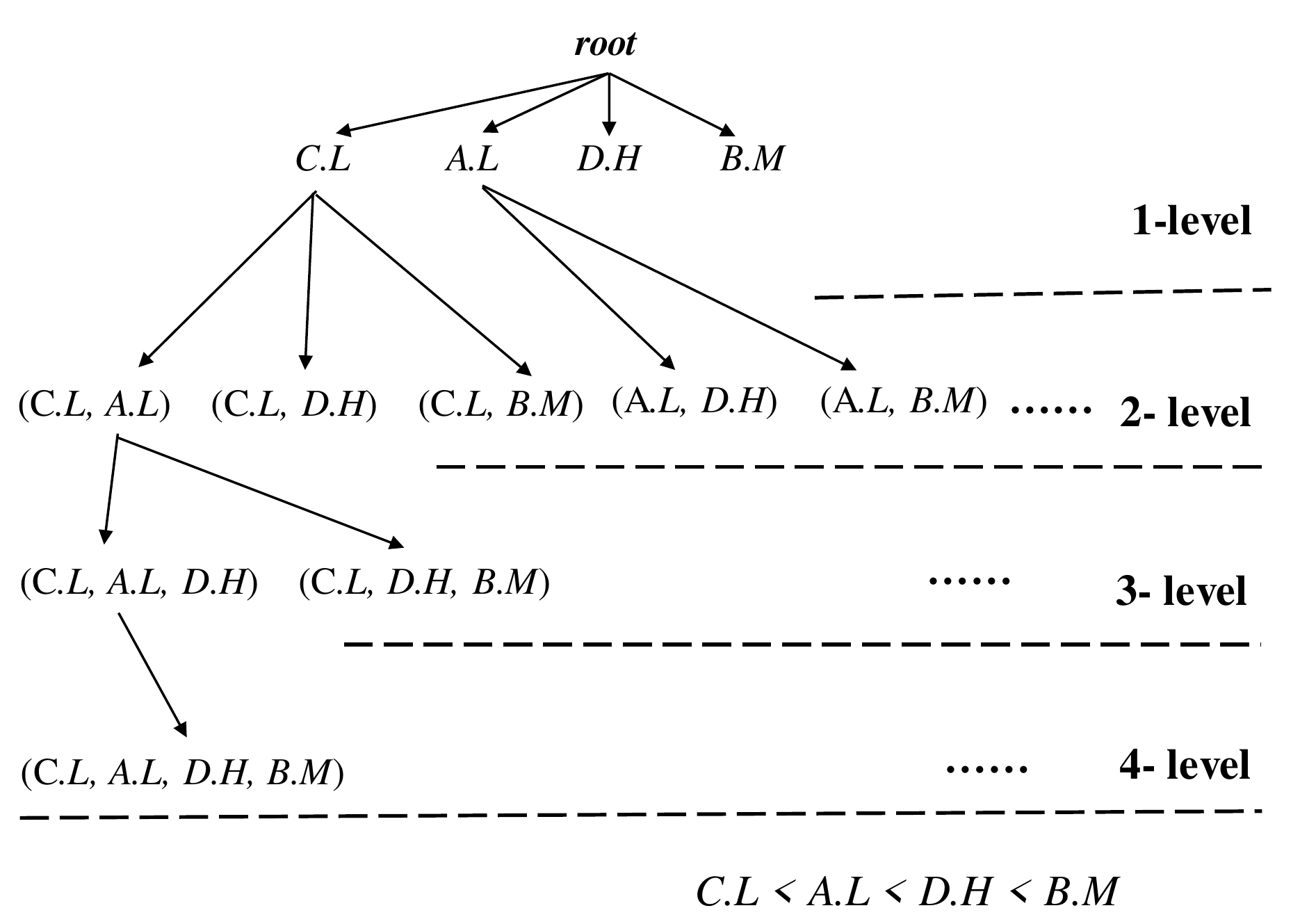}
	\captionsetup{justification=centering}
	\caption{An enumeration tree of the used example.}
	\label{fig:enumeratetree}
\end{figure}

In Figure \ref{fig:enumeratetree}, we narrow the search space by pruning to remove items below the minimum threshold. The remaining items are combined to form a new $k$-itemset. The structure of the fuzzy-list of the $k$-itemset is similar to that of the 1-itemset. The specific description is as follows:

\begin{theorem}
	\rm According to the concept of FRIs, for an itemset $ X $ in the fuzzy-list structure, if its  \textit{SUM(X.if)} is no less than the minimum fuzzy rare support and no more than the maximum fuzzy frequent support, it is seen as a FRI. If \textit{min}(\textit{SUM(X.if)}, \textit{SUM(X.rf)}) of $ X $ is no less than the minimum fuzzy rare support, new itemsets are required to be generated. If the sum of the resting fuzzy values of $ R_{il} $ is no less than the minimum fuzzy rare support (\textit{minRSup} $ \times$ $|Q'|$), extensions of $ R_{il} $ may be a FRI. If the summation of the resting fuzzy values of $ R_{il} $ is smaller than the minimum fuzzy rare support, any extensions of $ R_{il} $ will neither be a FRI nor a FFI. Thus, there is no need to construct a new fuzzy-list structure for its extension.
\end{theorem}

\begin{proof}	
		\rm	For $\forall t_{q}\supseteq X'$, suppose fuzzy term $ R_{il} $ is denoted as \textit{X}, and \textit{X}' is the extension of \textit{X} ($X\subset X' \subseteq t_{q} \Rightarrow X'.tids \subseteq X.tids$), thus $ (X'-X)$ = $(X'/X) $ and $(X'/X)$ $\subseteq(t_{q}/X) $. Since $if(X',t_{q})$ = $min\{if(X,t_{q})$, $rf(R_{il}$, $t_{q})\}$, it holds $if(X')$ = \textit{SUM(X.rf)}  \cite{lin2015fast}. Thus, with the definition of a FRI, both \textit{minRSup} and \textit{minFSup} are able to determine the promising candidates for FRIs by quickly pruning the search space.
\end{proof}

For example, we can consider the 3-itemsets (\textit{A.L}, \textit{D.H}, \textit{B.M}), which is an extension of 2-itemsets $ (A.L$, $D.H) $. Because the sum of the resting fuzzy values of $ (A.L$, $D.H) $ is calculated as (0.8 + 0 + 0.8 + 0.4) = 2.0. The extension $ (A.L$, $D.H$, $B.M) $ in the search space of FRI-Miner with respect to an enumeration tree \cite{rymon1992search} must be generated. For the 3-itemsets $ (C.L$, $A.L$, $D.H) $, which is the extension of 2-itemsets $ (C.L$, $A.L) $, the sum of the resting fuzzy values of $ (C.L$, $A.L) $ is calculated as 0.8 $ < $ 2.0. Thus, the extension of $ (C.L$, $A.L$, $D.H) $ is unnecessary because the fuzzy value of the 3-itemsets is lower than the specified minimum fuzzy value. A detailed description of the FRI-Miner algorithm is presented in the following section. The pseudocode for the construction of the fuzzy-list here is similar to that of the frequency-based FFI-Miner algorithm \cite{lin2015fast}.

\subsection{Fuzzy rare itemset mining phase}

Note that we have learned about both rare and frequent itemsets from the previous sections. In this section, we go into more details about the mining processes of FRIs. In general, there are three types of rare itemsets, as follows: 

\begin{enumerate}
	\item  The itemsets with only fuzzy rare items: all items in an itemset are rare, and the itemset may be rare.
	
	\item  The itemsets with fuzzy rare and fuzzy frequent items: an itemset that contains rare and frequent items internally may be a rare itemset. For example, we can calculate the sum of the inner fuzzy values of itemset  $ (A.L) $ is calculated as (= 2.8 $ > $ 2.0), and the sum of the inner fuzzy values of the itemset  $ (D.H) $ is calculated as (= 4.0 $\geq$ 4.0), but we can calculate the sum of the inner fuzzy values of the itemset  $ (A.L,D.H) $ is calculated as (= 2.0 $\geq$ 2.0).

	\item  The itemsets with only fuzzy frequent items: an itemset in which all items are frequent may be rare. For example, we can calculate the sum of the inner fuzzy values of the itemset $ (D.H) $ as 4.0. The sum of the inner fuzzy values of the itemset  $ (B.M) $ is calculated as (= 5.8 $ > $ 4.0), but the value of the itemset  $ (D.H,B.M) $ is calculated as (4.0 $ > $ 3.8 $ > $ 2.0).

\end{enumerate}

\begin{algorithm}[!hbtp]
	\caption{FRI-Miner} 
	\label{algorithm_1}
	\textbf{Input:} A quantitative database $D$; \textit{minRSup}; \textit{minFSup}.\\
	\textbf{Output:} \textit{FRIs}: all valid fuzzy rare itemsets.
	\begin{algorithmic}[1]
		\State First, it fuzzes the quantitative database to form a new fuzzy database, and then forms a corresponding fuzzy-list structure for each item.
		\For{ each fuzzy-list $X\in FLs$}
		\If{\textit{SUM(X.if)} $\geq$ $minRSup \times |D|$}
		\State $FLs \gets X \cup FLs $; // the candidates
		\If{\textit{SUM(X.if)} $\leq$ $minFSup \times |D|$}
		\State $FRIs \gets X \cup FRIs$; // output the results
		\EndIf
		\EndIf
		\If{\textit{SUM(X.rf)} $\geq$ $minRSup \times |D|$}
		\State initialize $exFLs \gets null $;
		\For {each fuzzy-list $Y$ after $X \in FLs$}
		\State call the construct($X$,$Y$) procedure, then put the built fuzzy-list into \textit{exFLs};
		\EndFor
		\State call FRI-Miner (\textit{exFLs}, \textit{minRSup});
		\EndIf
		\EndFor
		\State return \textit{FRIs}		
	\end{algorithmic}
\end{algorithm}

In the above discussion, we have carried out detailed concepts, data structure, theorem, and strategies of the proposed FRI-Miner algorithm. Algorithm \ref{algorithm_1}  describes the complete details of FRI-Miner for discovering FRIs. Notice that the itemset that satisfies  \textit{SUM(X.if)} $\geq$ $minRSup \times |D|$ will be rare or frequent (lines 2-3), and it is used to generate the new itemsets (extensions of $X$,  lines 11-12). This guarantees the correctness and completeness of the final discovered results of FRI-Miner.

According to the above pseudo-code, it can be known that the quantitative database as follows shown in Table \ref{table:db} is converted to a fuzzy database as follows shown in Table \ref{table:tansformed} at the beginning of the algorithm, and the membership function $\mu$ is used to blur the quantitative value at this phase. The fuzzy items in the fuzzy database obtained in the first phase are clearly planned in the fuzzy-list structure, and the fuzzy value that meets the user-specified threshold is selected by the maximum scalar cardinality and support-ascending order strategy. It keeps the rare and frequent items that meet the conditions and outputs rare items. It combines the remaining fuzzy items in a tree structure and deletes the itemsets that do not meet the conditions through the pruning strategy. Finally, it outputs complete FRIs that meet the conditions. 
\section{Experimental Evaluation}
 \label{sec:experiments}
 
In this study, we address the problem of mining FRIs in a quantitative database. To the best of our knowledge, FRI-Miner is the first fuzzy-based algorithm for mining rare itemsets with linguistic meanings. It can discover rare itemsets that satisfy the two conditions of fuzzy values. The RP-Growth algorithm \cite{tsang2011rp}, which aims to mine rare itemsets, was selected as the baseline to evaluate the validity and performance of the proposed FRI-Miner algorithm.

A total of six benchmark datasets were used in this experiment: retail, accidents, foodmart, chess, mushroom, and kosarak. Real-life datasets are available in a public repository\footnote{\url{http://www.philippe-fournier-viger.com/spmf/}}. Their characteristics and descriptions can be found in previous studies \cite{lin2015rwfim,lin2017fdhup}. All algorithms were implemented using Java language and executed on an Intel Core 6 Duo 3.00 GHz machine with 8 GB of RAM running Windows 10.

We fuzz the quantitative datasets according to a predefined membership function. For the foodmart dataset, we used a different membership function, which is [100, 600, 1100]. The reason is to ensure that the data in the dataset are evenly distributed in the set membership function, so that the final result obtained is more accurate. The other datasets were distributed between the specified [1, 21, 31]. According to the previous algorithm, to evaluate the performance of the compared algorithms, we analyzed them from three aspects: running time, memory usage, and the number of generated patterns.  For example, the changes in runtime can reflect whether the designed data mining algorithm is acceptable within a reasonable execution time. The results obtained with the compared algorithms are shown in Figure \ref{datasets Runtime}, Figure \ref{RTMIN}, Table \ref{table:pattern},  and Figure \ref{PMAX}, respectively. The detailed results are as follows.

\subsection{Effectiveness}

As mentioned before, FPM and RPM are two different mining tasks. Moreover, the fuzzy-theoretic-based RPM is also different from the traditional RPM.  In the mining rare itemsets, two thresholds \textit{minRSup} and \textit{minFSup} are set in FRI-Miner. As shown in Table \ref{table:pattern} and Figure \ref{datasets Runtime}, the constraints of the two thresholds further enable the FRI-Miner algorithm to run efficiently, shorten the running time, and identify meaningful rare itemsets. First, FRI-Miner uses \textit{minRSup} to remove itemsets that are lower than \textit{minRSup}; that is, these itemsets without practical significance. Then, fuzzy-based rare itemsets are discovered from the potential candidates. In the experiments, \textit{minRSup} and \textit{minFSup} are expressed as \textit{minSup} and \textit{maxSup}, respectively, as shown in Table \ref{table:pattern}. Note that the number of \#\textit{FRIs} is derived by FRI-Miner, and the number of \#\textit{RIs} is derived by the RP-Growth.

\begin{table*}[ht]
	
	\fontsize{6.8pt}{9pt}\selectfont
	\centering
	\caption{Number of patterns (candidates and final results) under various parameter settings}
	\label{table:pattern}
	\begin{tabular}{|cc|lllllll|}
		\hline\hline
		\multirow{2}*{\textbf{}}&
		\multirow{2}*{\textbf{}}
		&\multicolumn{7}{c|}{\textbf{\# of patterns under various parameter settings}}\\
		\cline{3-9} 
		&& \textit{test}$_1$ & \textit{test}$_2$ & \textit{test}$_3$ &  \textit{test}$_4$ &  \textit{test}$_5$  &  \textit{test}$_6$ &  \textit{test}$_7$ \\ \hline

		&  \textbf{\textit{minSup}} & 100 &	100 & 100 &	100 & 200 &	200& 200 \\
		&  \textbf{\textit{maxSup}} & 250 &	300 & 350 &	400 & 250 &  300& 350 \\
		(a) accidents  &  \textbf{\#\textit{FRIs}} & 6,638 & 10,791 &  17,966 & 29,032 & 837   & 1,516 & 2,293  \\
		& \textbf{\#\textit{RIs}} & 387,624 & 539,064 & 2,924,992 & 3,697,552 &	500 & 1,656 & 28,868 	\\
		\hline

		&  \textbf{\textit{minSup}} & 5 &	5 &	5 &	10 & 10 &	10 & 10  \\
		&  \textbf{\textit{maxSup}} & 13 & 	15 & 17 & 13 & 15 &	17 & 20 \\
		(b) foodmart  &  \textbf{\#\textit{FRIs}} & 4,514 &  7,905 & 9,525 & 4,263 & 7,654 &	9,274 & 10,031  \\
		& \textbf{\#\textit{RIs}} & 930 & 1,223 & 1,402 & 554 & 847 & 1,026 & 1,145	\\
		\hline

		&  \textbf{\textit{minSup}} & 100 &	100 & 100 &	150 & 150 &	150 & 150 \\
		&  \textbf{\textit{maxSup}} & 200 &	300 & 400 &	200 & 300 &	400 & 500  \\
		(c) chess &  \textbf{\#\textit{FRIs}} & 1,448 & 18,780 & 99,014 & 	187 & 	2,013 & 12,886 & 59,374  \\
		& \textbf{\#\textit{RIs}} & 16,414,992 & 50,739,408 &  675,273,704 & 122,040 & 1,309,176 & 	64,897,692 & 237,222,940	\\
		\hline

		&  \textbf{\textit{minSup}} & 1000 & 1000 &	1000 &	1500 &	1500 &	1500 & 1500 \\
		&  \textbf{\textit{maxSup}} & 3000 & 3500 &	4000 &	3000 &	3500 &	4000 & 4500 \\
		(d) kosarak &  \textbf{\#\textit{FRIs}} & 2,242 & 2,395 & 2,546 & 801 & 	882 & 	977 & 1,052  \\
		& \textbf{\#\textit{RIs}} & 642,536 & 654,922 & 661,168 & 195,507 & 198,575 & 201,007 & 202,013	\\
		\hline
		
		&  \textbf{\textit{minSup}} & 100 &	100 &	100 &	150 &	150 &	150 & 150 \\
		&  \textbf{\textit{maxSup}} & 200 &	300 &	400 &	200 &	300 &	400 & 500 \\
		(e) mushroom  &  \textbf{\#\textit{FRIs}} & 1,344 & 5,862 & 15,405 & 	243 & 	1,185 & 3,062 & 5,185  \\
		& \textbf{\#\textit{RIs}} & 364,992 & 1,269,244 & 1,915,464 & 360,512 & 	475,200 & 	882,000 & 2,383,752 	\\
		\hline
		
		&  \textbf{\textit{minSup}} & 95 &	95 & 95 & 95 &	100 &	100 &   100  \\
		&  \textbf{\textit{maxSup}} & 200 &	300 & 400 & 500 &	200 &	300 &   400  \\
		(f) retail &  \textbf{\#\textit{FRIs}} & 1,261 & 1,581 & 1,734 & 1,795 & 1,128 & 1,441 &  1,593 \\
		& \textbf{\#\textit{RIs}} & 1,740 & 3,244 & 3,887 &	4,378 &	1,493 & 2,929 &  3,552 	\\
		\hline
		
		\hline\hline
	\end{tabular}
\end{table*}

\begin{figure*}[ht]
	\centering
	\includegraphics[scale=0.4]{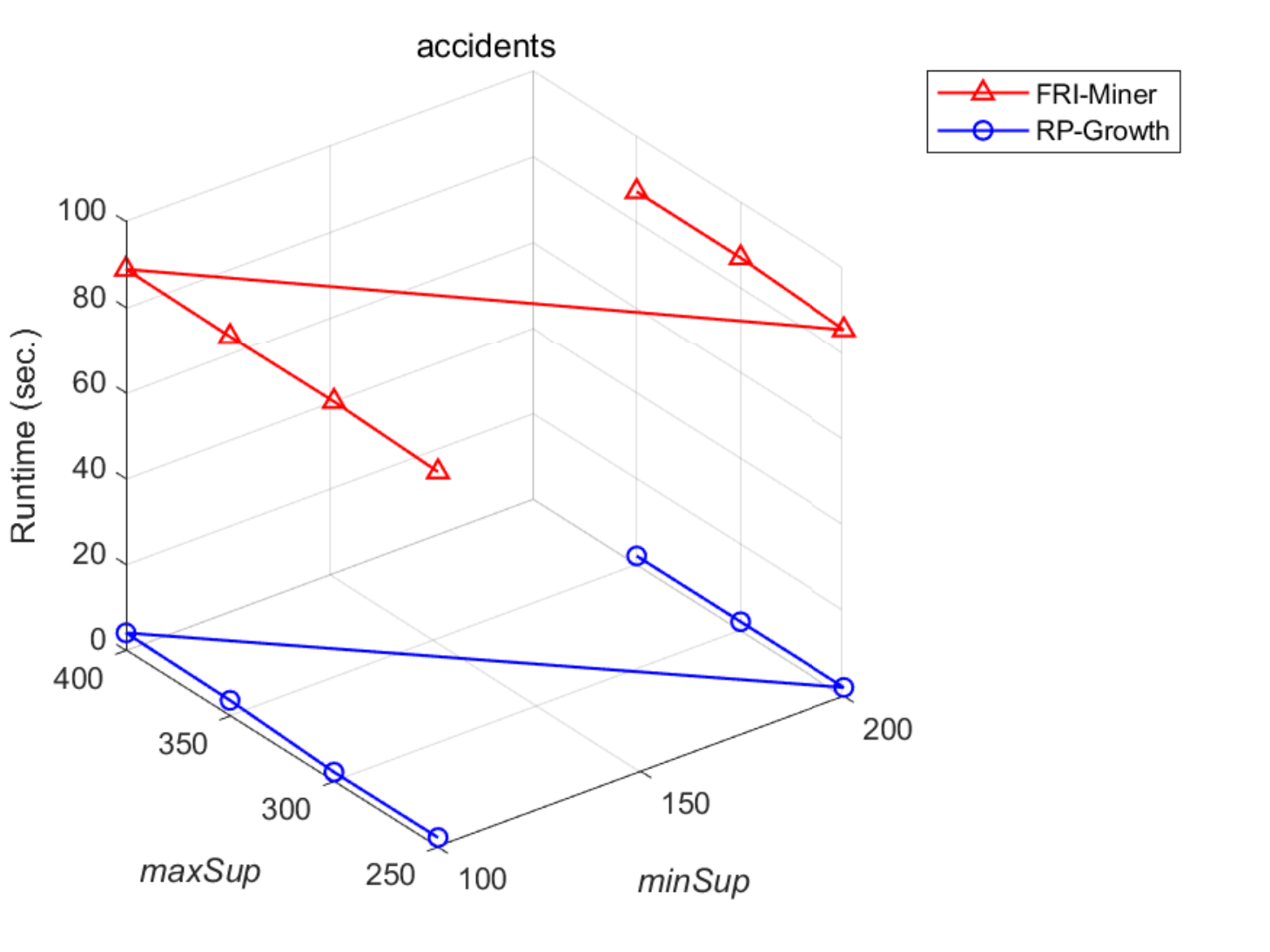}
	\includegraphics[scale=0.4]{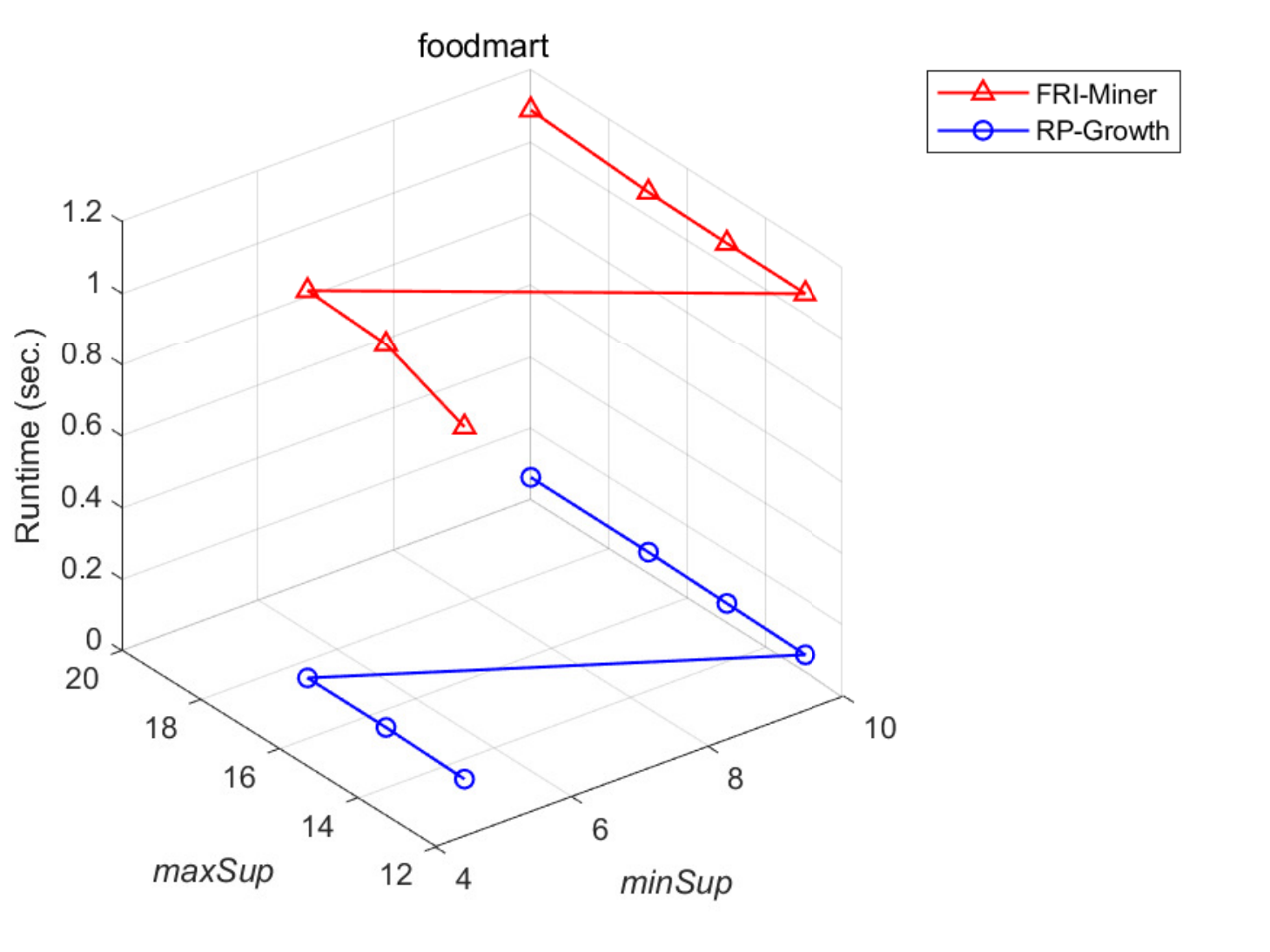} 
	\includegraphics[scale=0.4]{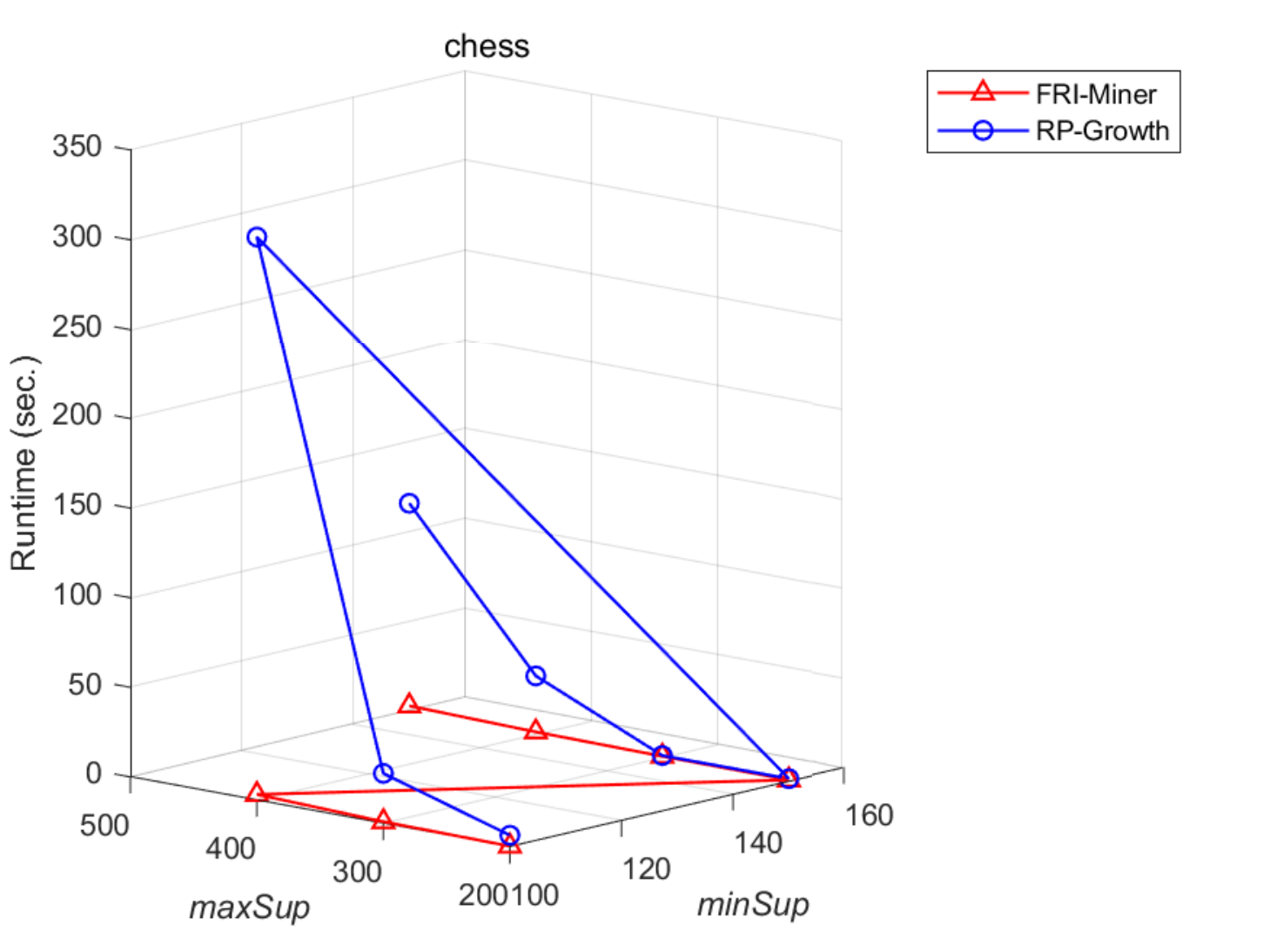}	
	\includegraphics[scale=0.4]{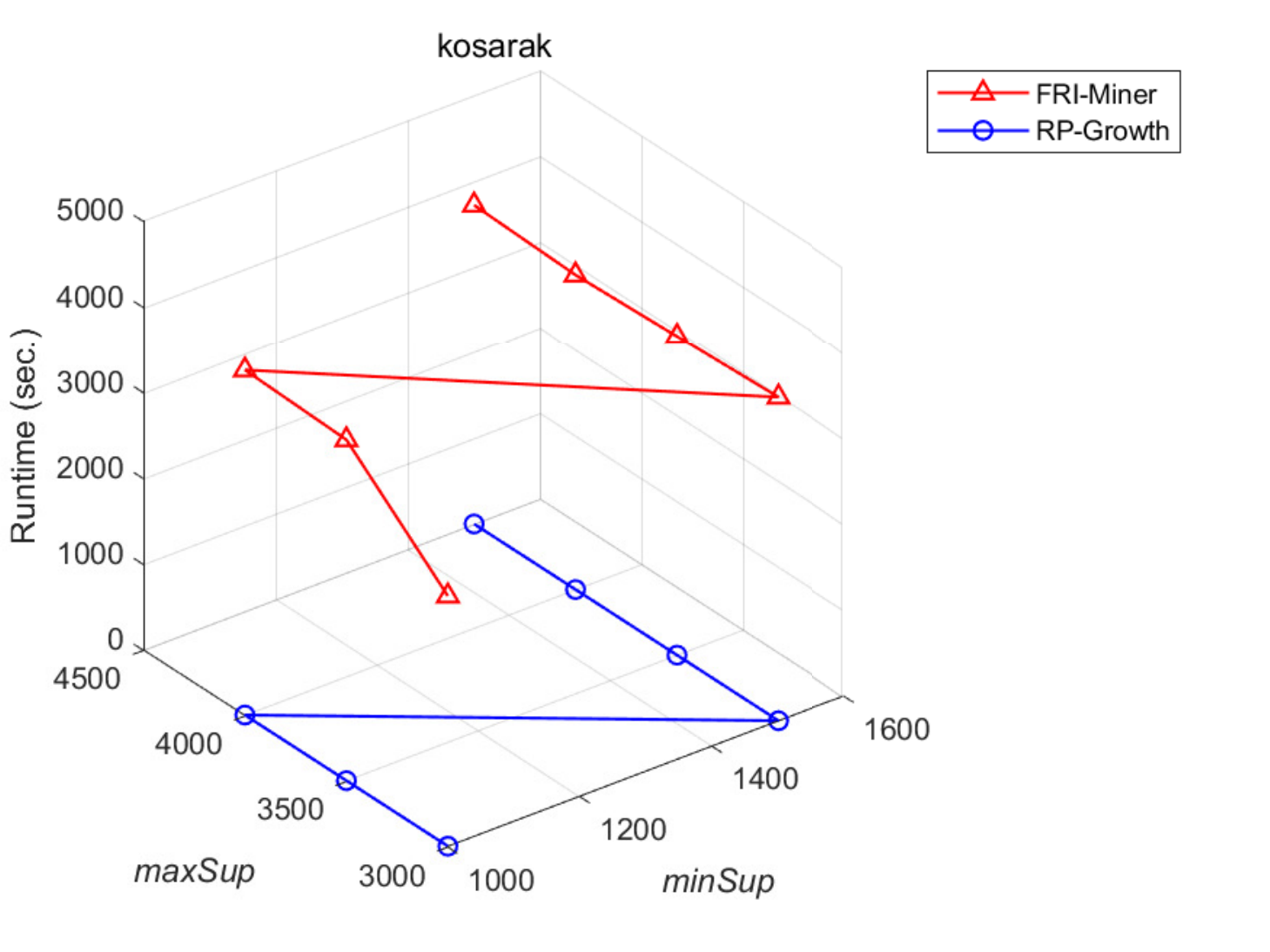}
	\includegraphics[scale=0.4]{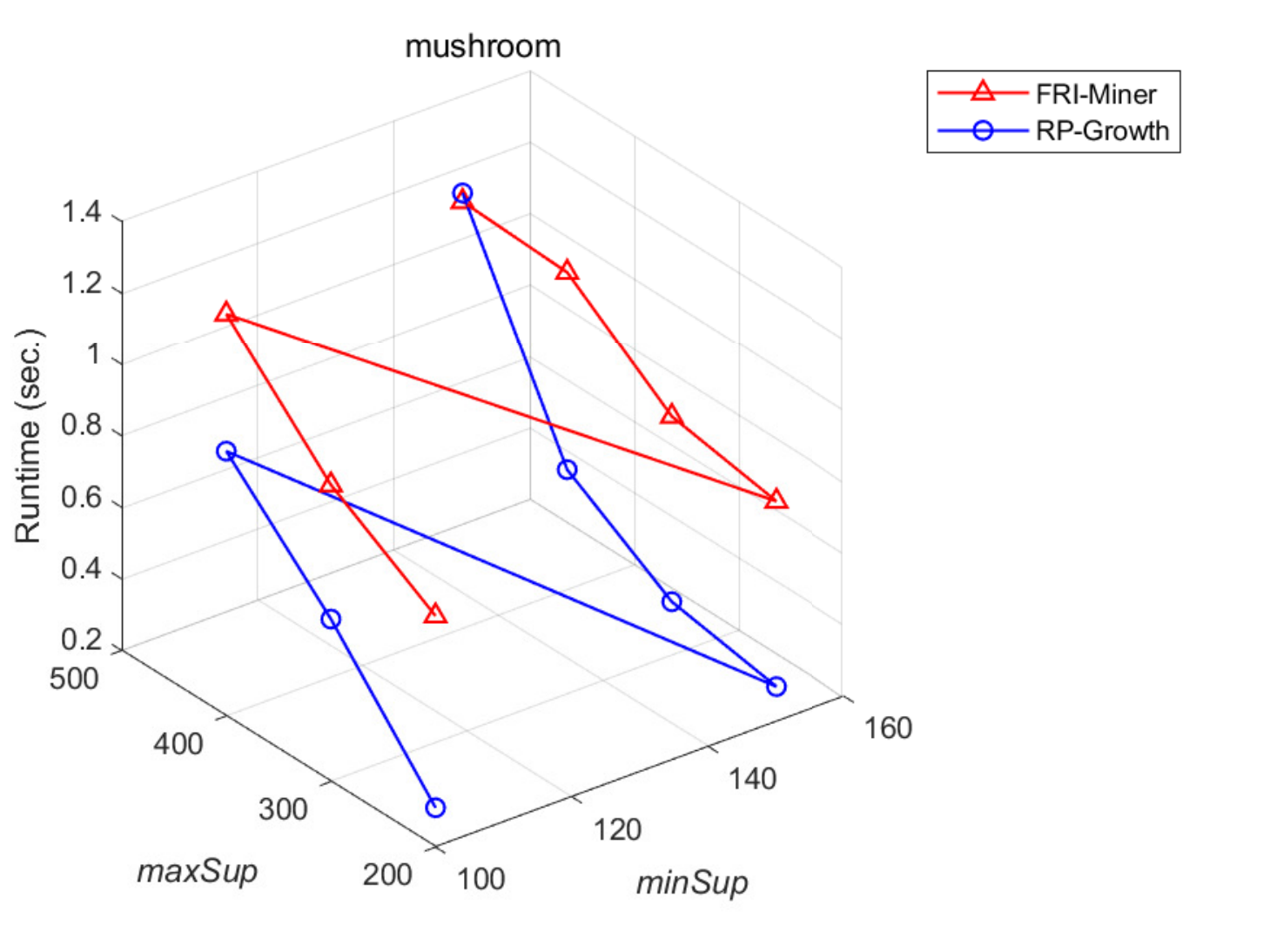}
	\includegraphics[scale=0.4]{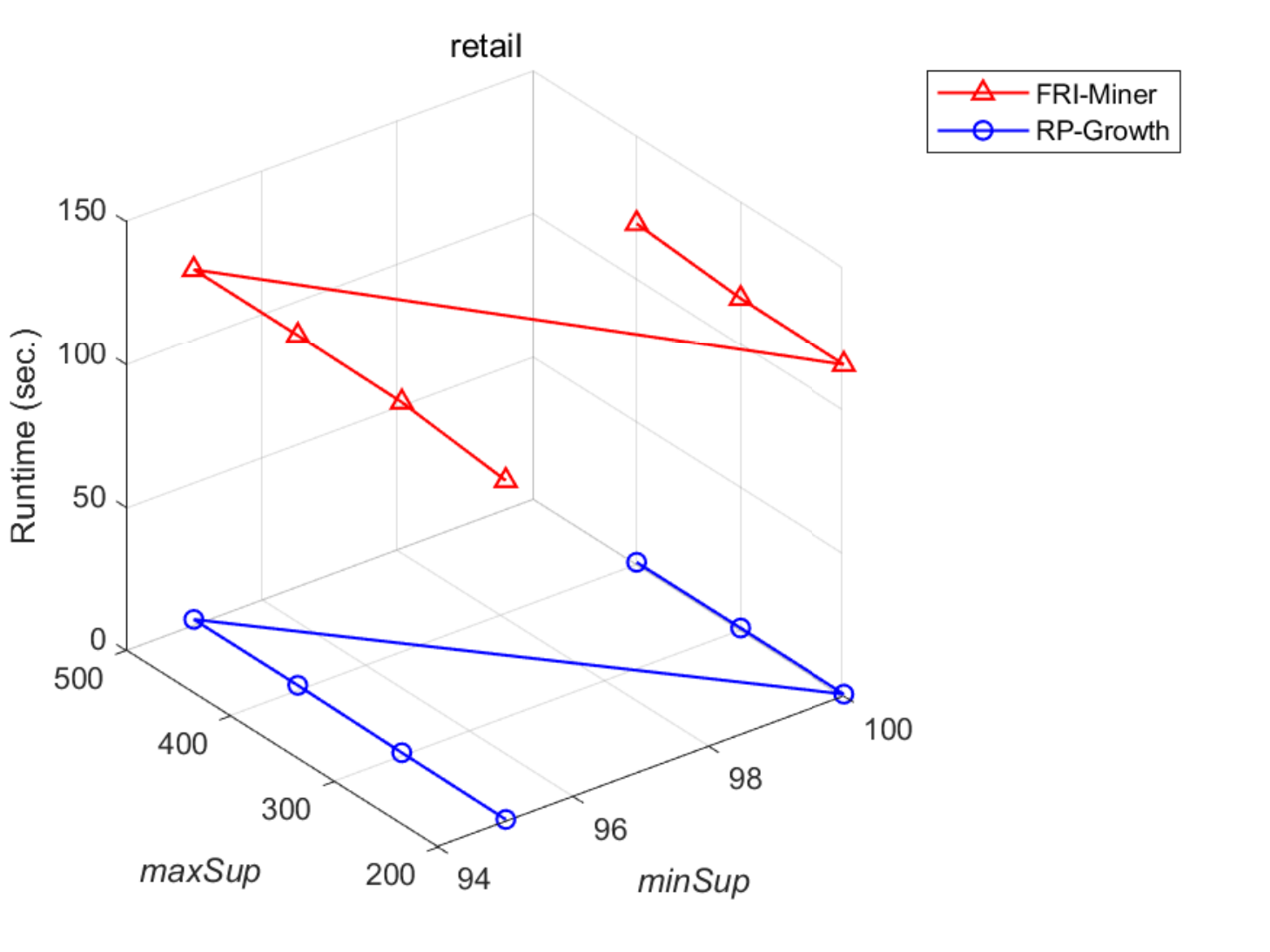}
	
	\caption{Runtime vs. \textit{minSup}/\textit{maxSup}}
	\label{datasets Runtime}
\end{figure*}

\begin{figure*}[!htbp]
	\centering
	\includegraphics[trim=90 0 0 0,clip,scale=0.63]{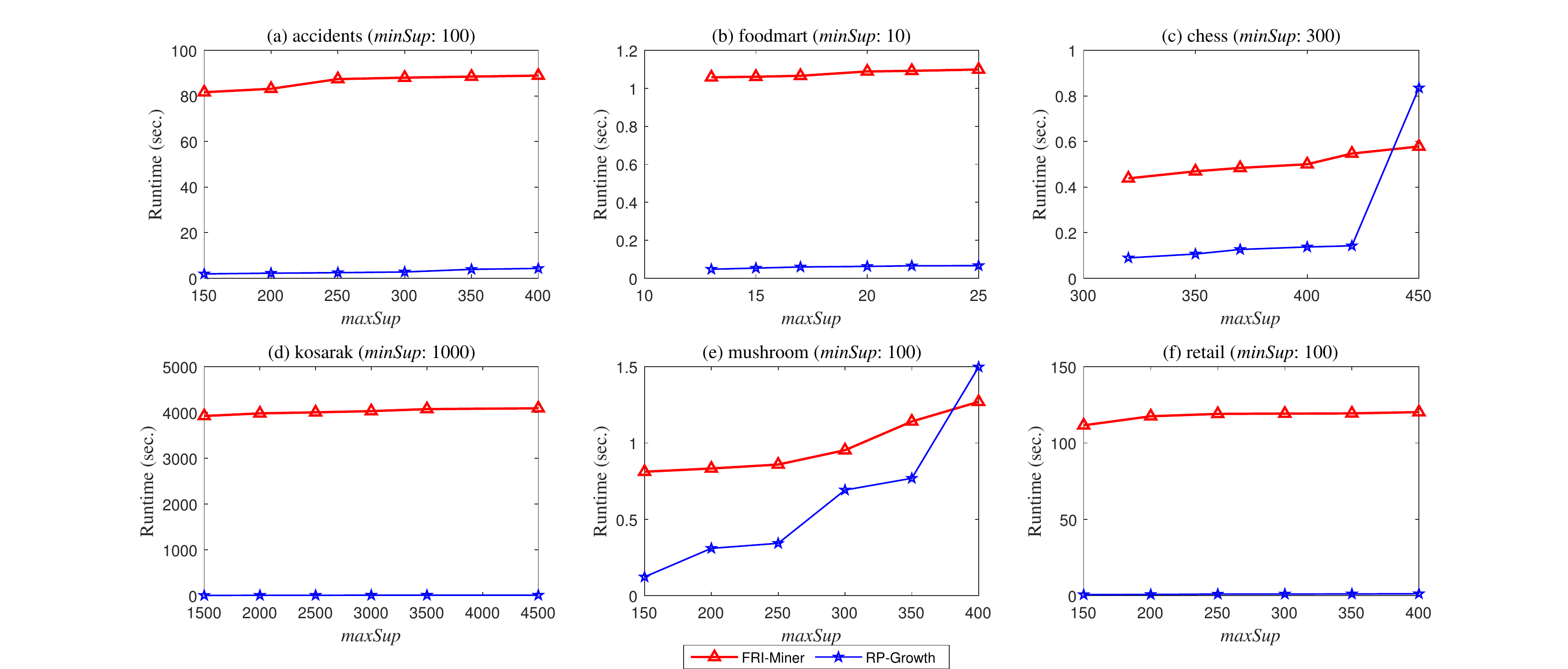}
	\caption{Runtime vs. \textit{maxSup}}
	\label{RTMAX}
\end{figure*}

\begin{figure*}[!htbp]
	\centering
	\includegraphics[trim=90 0 0 0,clip,scale=0.63]{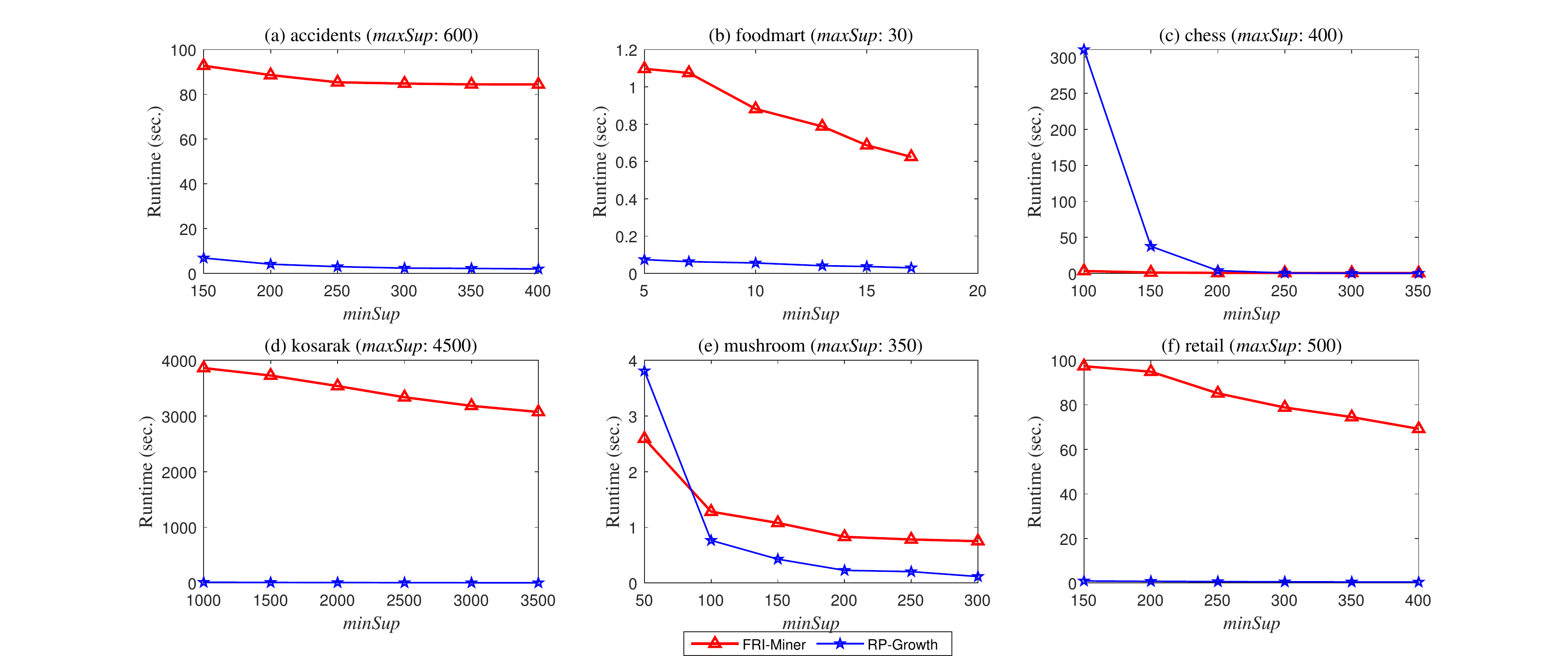}
	\caption{Runtime vs. \textit{minSup}}
	\label{RTMIN}
\end{figure*}

In general, in the experimental process, the running time reflects the execution efficiency of the algorithm. According to the experimental results, we find that the running time of the FRI-Miner algorithm is relatively constantly shortened, and the number of discovered patterns also increases. This RP-tree-based RP-Growth algorithm is also a prominent algorithm in mining rare itemsets, and FRI-Miner is also based on its efficient performance. 

It can be observed that the setting of two thresholds also affects the mining time of the itemset.  In addition, normal behavior/patterns often appear in the form of  infrequent itemsets, but fewer rare behaviors/patterns can be found with FRI-Miner using fuzzy theory. For example, consider the retail dataset, when setting parameters as \textit{minSup} from 95 to 100, and \textit{maxSup} from 200 to 500, the number of FRIs is changed from 1,261 to 1,593, while the number of RIs can reach 1,740 to 3,552.  It is clear that the two types of patterns, FRIs and RIs, are different, and the fuzzy-based patterns have more useful meaning. Thus, the concept of FRI has excellent fuzzy modeling capabilities and has a linguistic meaning.

\subsection{Efficiency w.r.t. runtime}

\textbf{Runtime vs. support}. The effects of the \textit{minSup} and \textit{maxSup} thresholds were evaluated first.  For each dataset, we uniquely specified the corresponding  thresholds. When the minimum support threshold we choose continues to increase and the maximum support threshold remains unchanged, we find that the running time of FRI-Miner will continue to decrease. The results of the detailed experiments are shown in Figure \ref{RTMIN}(a) to \ref{RTMIN}(f). Similarly, when the minimum support threshold is unchanged and the maximum support threshold continuously increases, the running time continuously increases as the mining range increases. These results are shown in Figure \ref{RTMAX}(a)-\ref{RTMAX}(f). For RP-Growth previously studied, it is a mining rare itemsets algorithm based on the RP-tree structure.  This is especially true when \textit{minSup} is set to be very small, and the runtime increases sharply. For example, in Figure \ref{RTMIN} and Figure \ref{RTMAX}, the running time of the two compared algorithms always increases with an increase in \textit{maxSup}, and the running time of the two compared algorithms always increases with a decrease in \textit{minSup}. This result corresponds to the following experiments.

\textbf{Runtime vs. size/density of dataset}. After the runtime analysis on the benchmark datasets, we found that the size of the tested dataset also had a significant impact on the running time of the experiments. For a larger dataset, we need to spend more time on the mining process. In addition, we found that, for a dense dataset, although the dataset is not large, it takes a longer time in the mining process.

\textbf{Discussion}. Based on the results of the datasets in the experimental run, we found that the running time of the FRI-Miner algorithm was longer than that of the RP-Growth algorithm. The reason is that in the process of mining rare itemsets, we need to fuzzify the quantitative values; thus, the running time is relatively longer. Another reason for the longer running time is that FRI-Miner considers that the combination of frequent itemsets may also be rare itemsets when mining rare itemsets. In other words, the branches of frequent itemsets are not removed in the pruning stage, which makes we find more rare itemsets that are ignored. Owing to the existence of pruning strategies,  FRI-Miner not only reduces the search space, but also improves the efficiency of the mining task.

\subsection{Efficiency w.r.t. pattern}

\begin{figure*}[!htbp]	
	\centering	
	\includegraphics[trim=90 0 0 0,clip,scale=0.65]{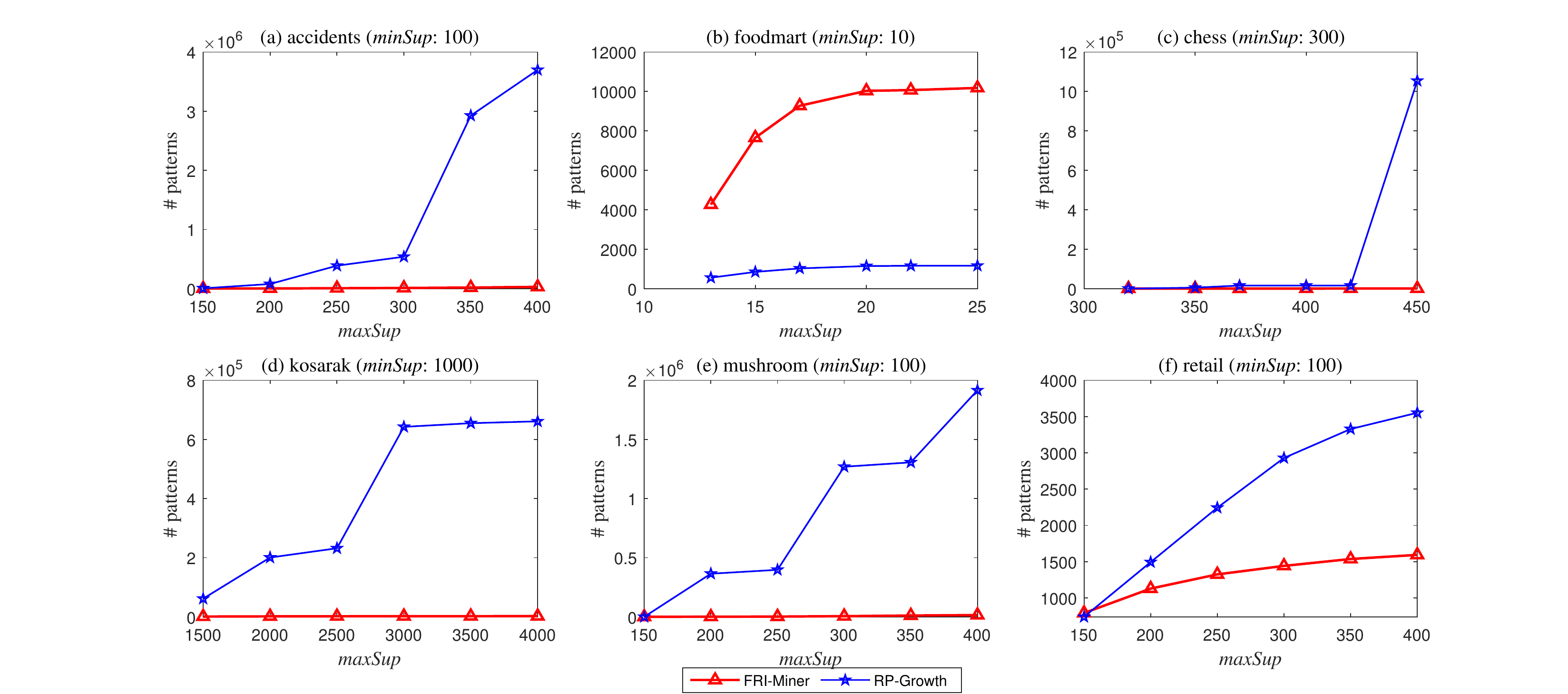}
	\caption{Pattern vs. \textit{maxSup}}
	\label{PMAX}
\end{figure*}	

\begin{figure*}[!htbp]		
	\hspace{0in}
	\centering
	\includegraphics[trim=90 0 0 0,clip,scale=0.63]{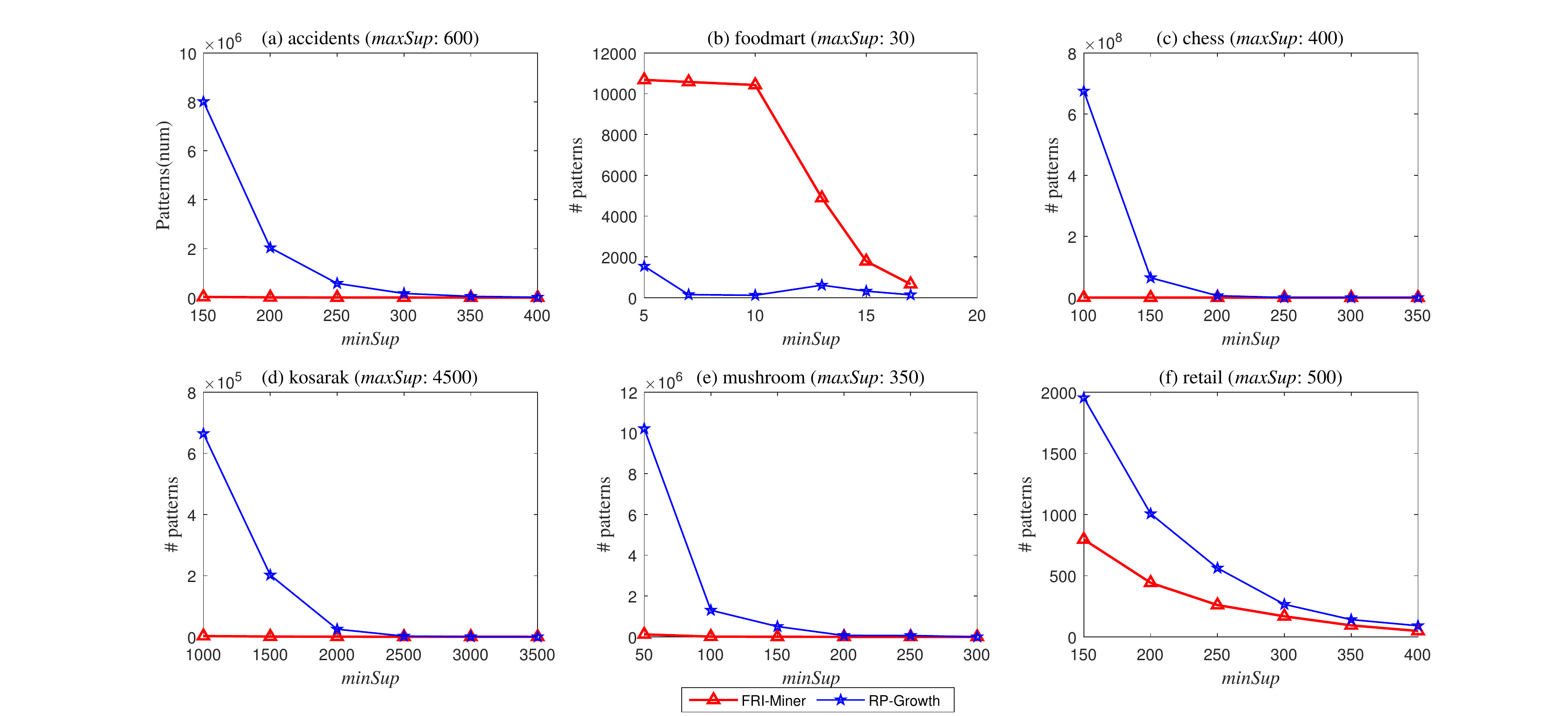}
	\caption{Pattern vs. \textit{minSup}}
	\label{PMIN}
\end{figure*}

\begin{figure*}[!htbp]
	\centering	
	\includegraphics[scale=0.4]{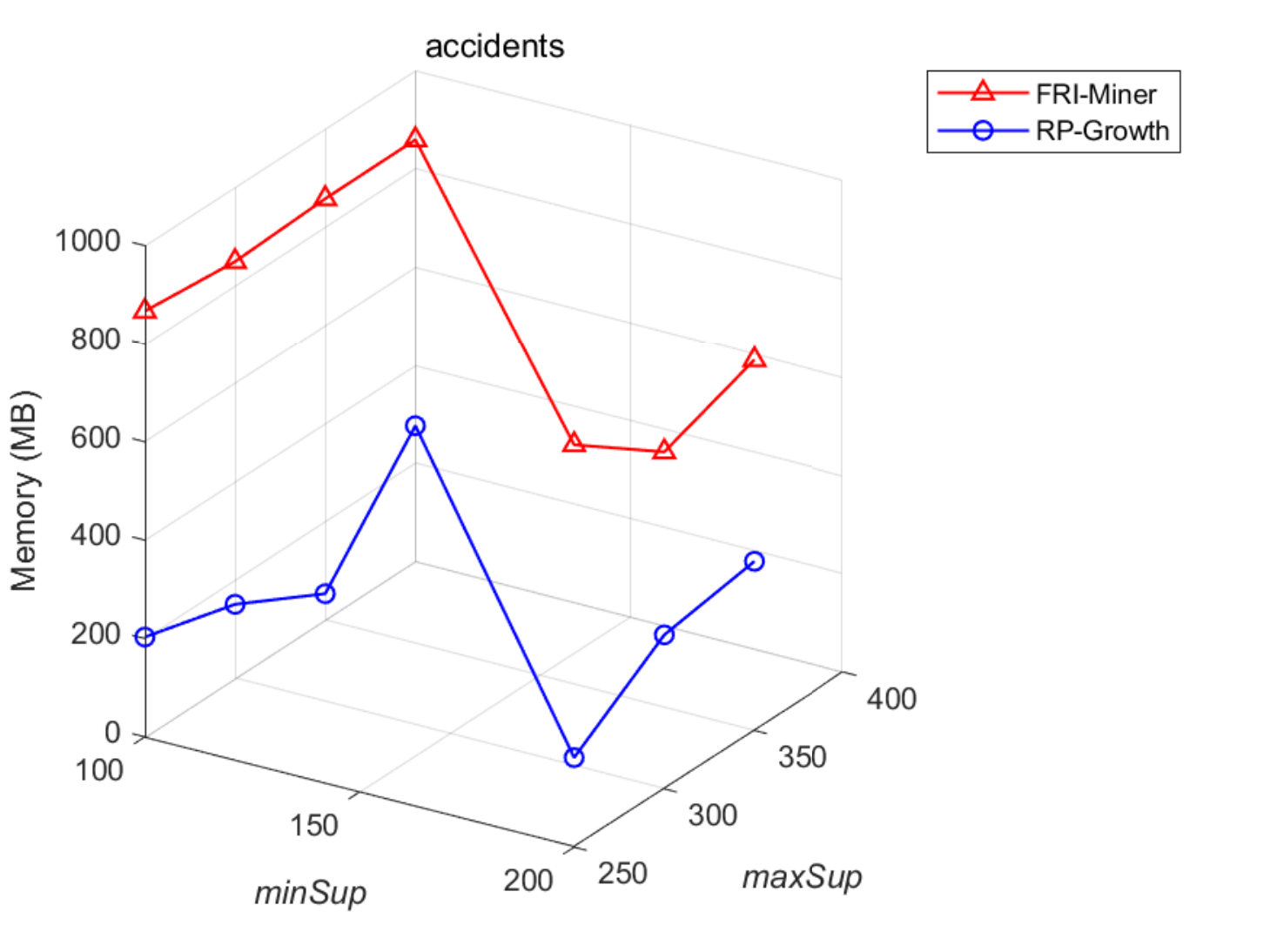}	
	\includegraphics[scale=0.4]{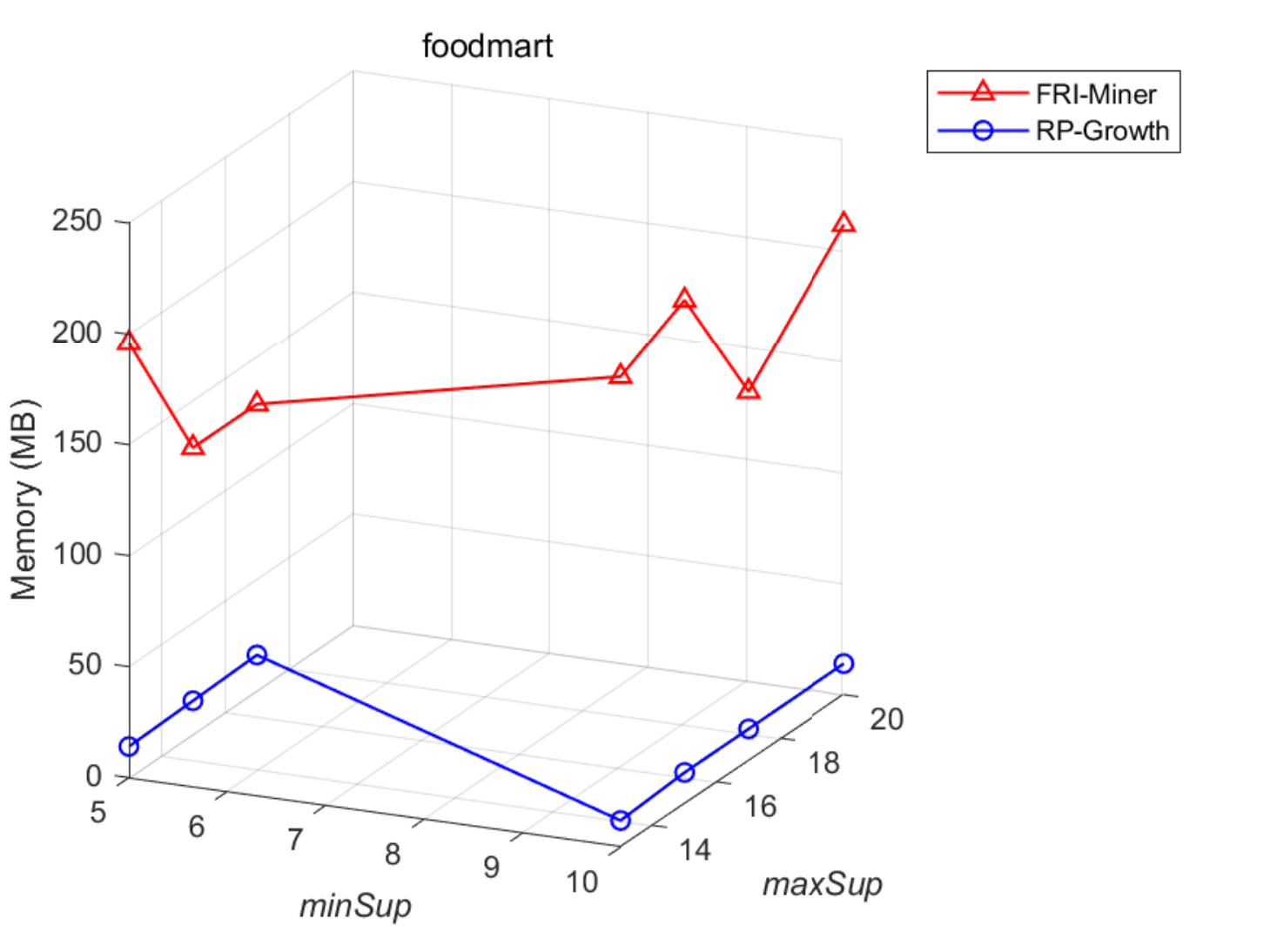}	 
	\includegraphics[scale=0.4]{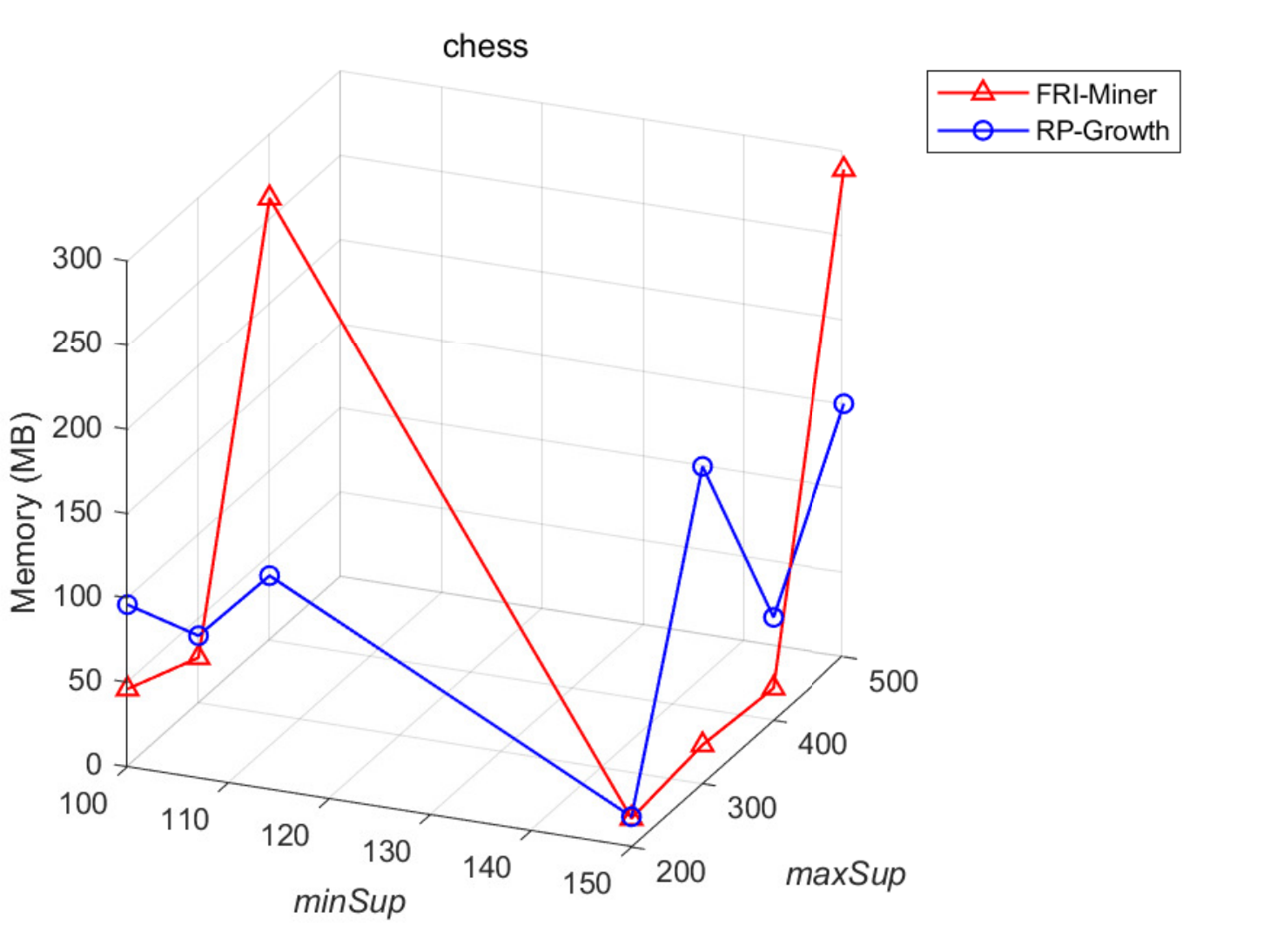}
	
	\includegraphics[scale=0.4]{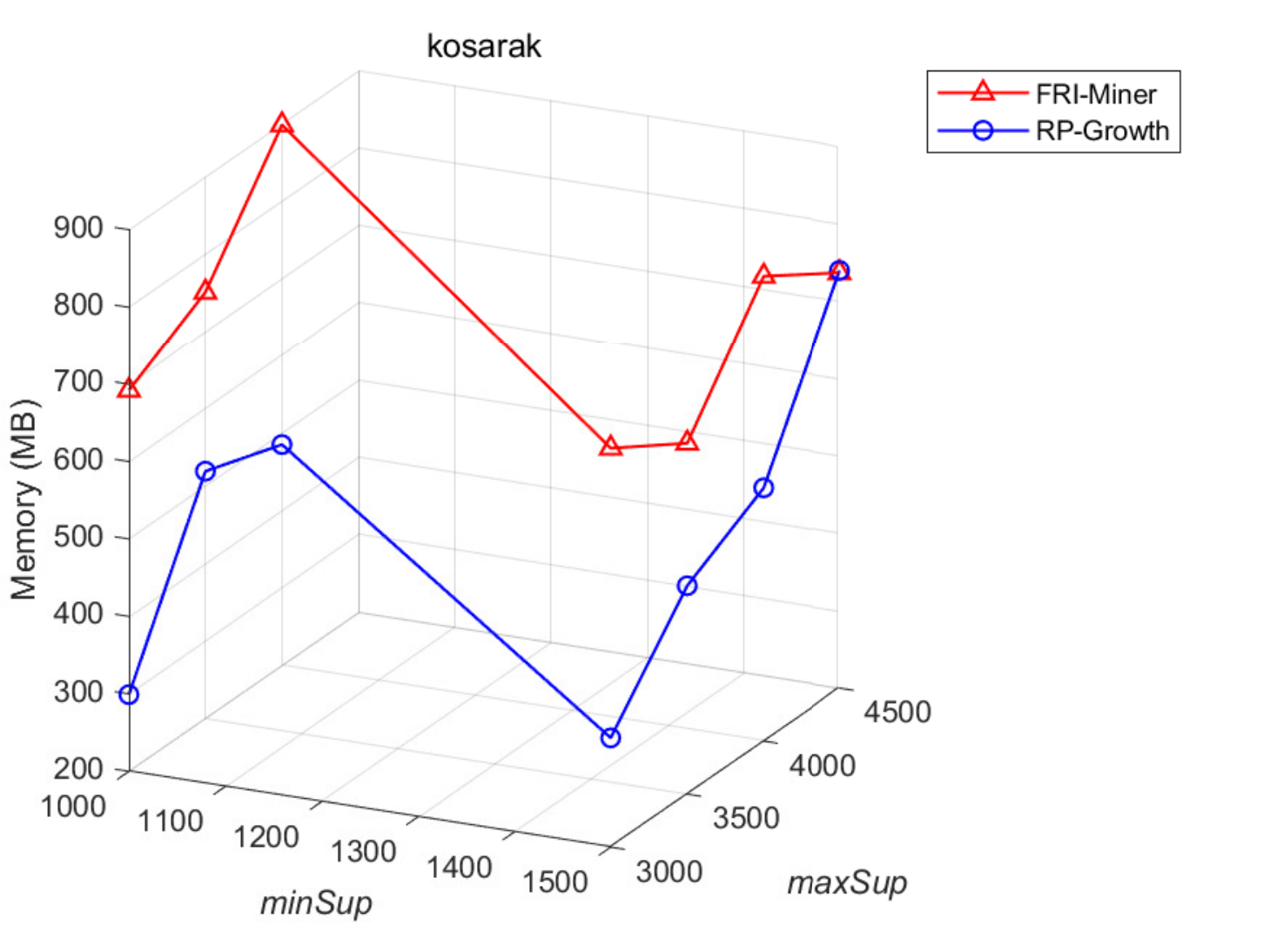}
	\includegraphics[scale=0.4]{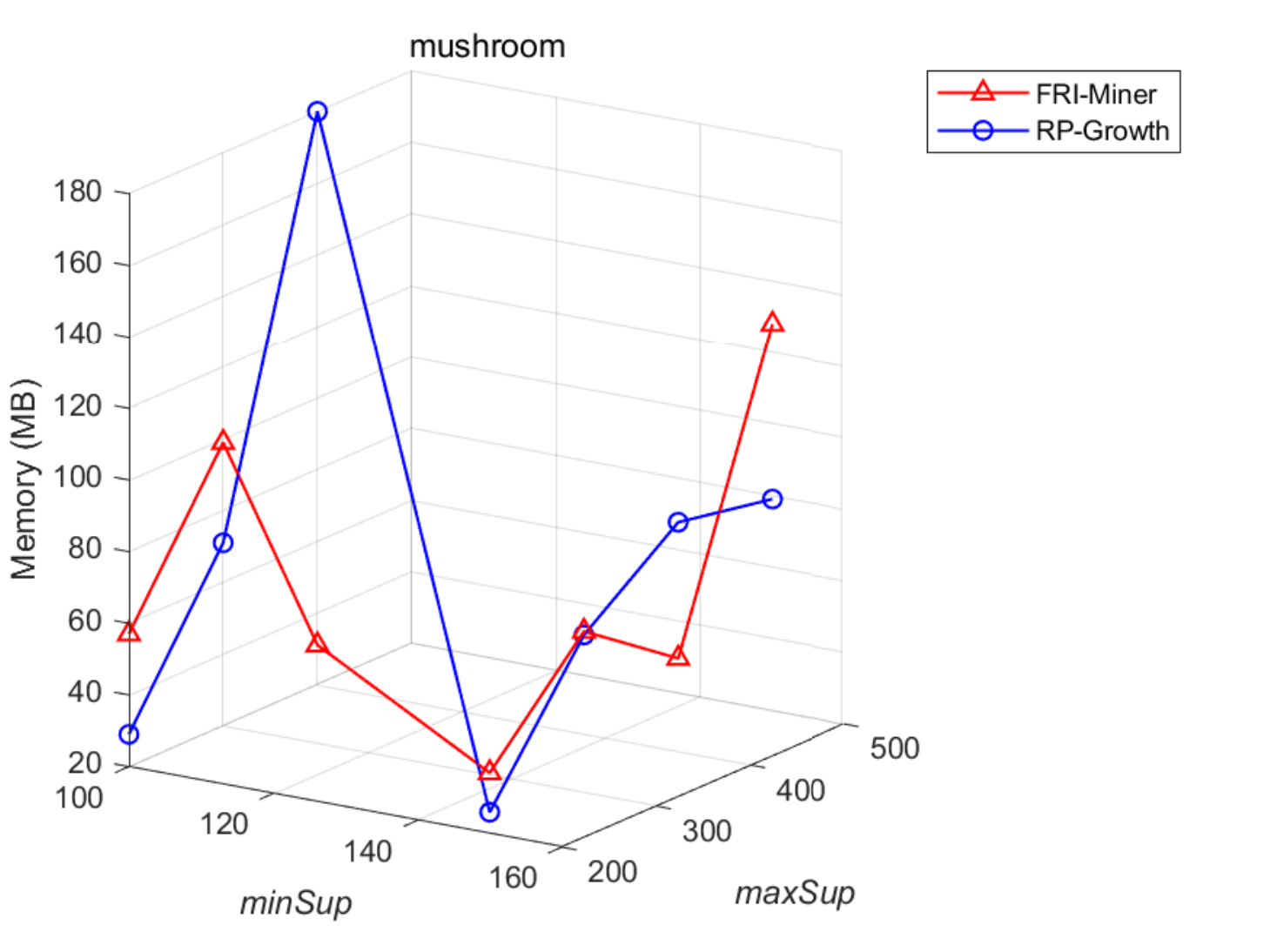}
	\includegraphics[scale=0.4]{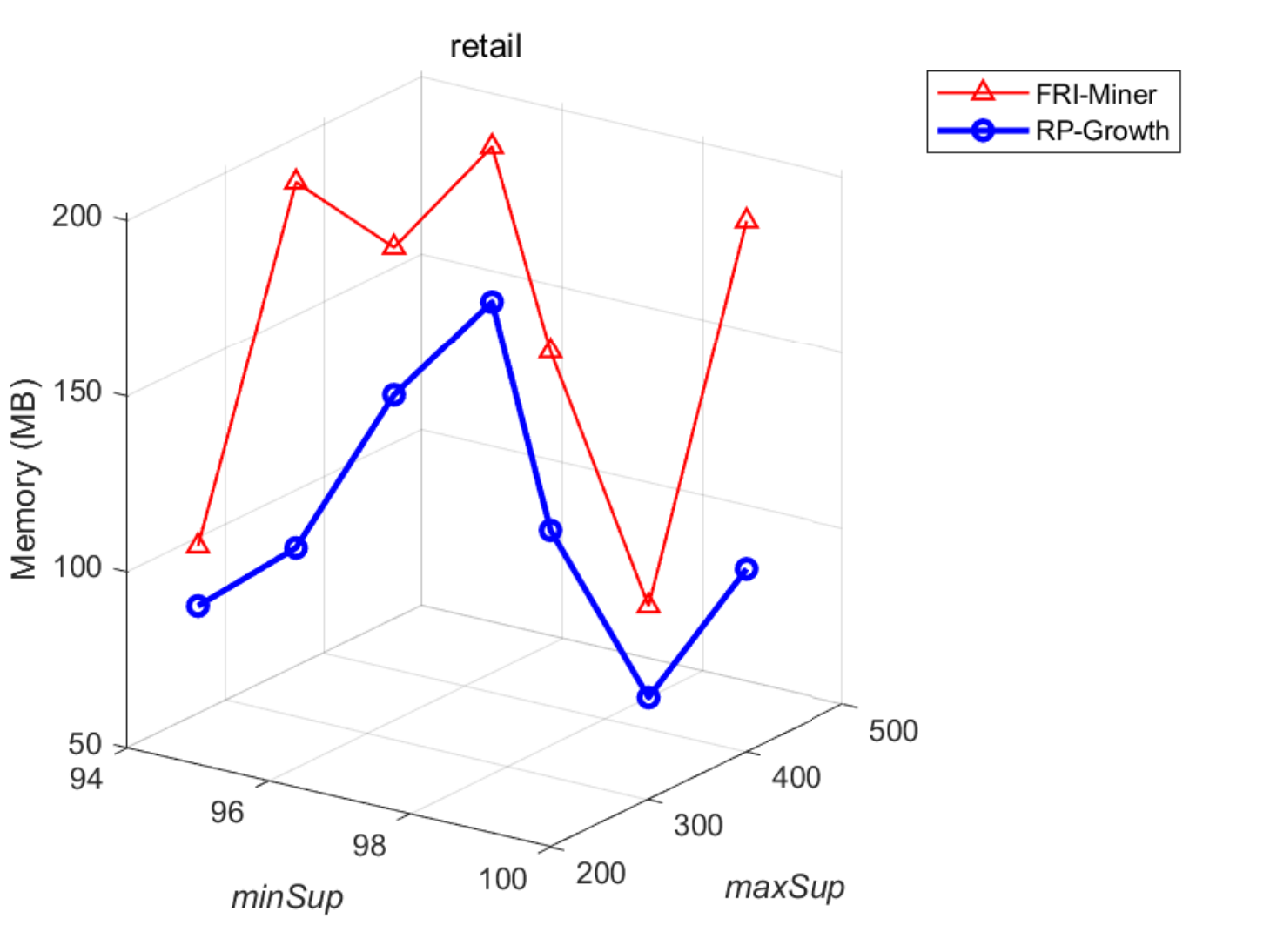}
	\caption{Memory vs. \textit{minSup}/\textit{maxSup}}
	\label{databases Memory}
\end{figure*}

\textbf{Pattern}. In this subsection, we compare the number of patterns discovered by the algorithms. Our intuition suggests that in the FRI-Miner algorithm, the number of patterns mined is greater than that of the RP-Growth algorithm. This has also been confirmed in some datasets in the experiment, such as dataset accidents and BMS, which well reflect this phenomenon. The results of the experiment may be unclear or conflicting . The reason for this experimental phenomenon is that we use quantitative values and thus use fuzzification and pruning strategies to prune many inconsistent itemsets in the mining process. This pruning strategy ensures that the number of itemsets we finally discover is relatively small, but the practicability of the itemsets that we discover is higher. In the experiments, we found that most datasets reflect this phenomenon, including retail, chess, mushroom, and kosarak. The specific experimental results are listed in Table \ref{table:pattern}.

\textbf{Pattern vs. support}. For each dataset, we choose a threshold that meets their own thresholds; when the minimum support threshold selected continues to increase and the maximum support threshold remains unchanged, we find that the number of itemsets extracted will continue to decrease. Figure \ref{PMIN} shows the experimental results. Similarly, when the minimum threshold remains unchanged and the maximum threshold value continuously increases, we find that the number of itemsets mined continuously increases. The experimental results are presented in Figure \ref{PMAX}. This is especially true when \textit{minSup} is set to be very small, and the runtime increases sharply. The effects of the \textit{minSup} and \textit{maxSup} thresholds were evaluated first. For example, in Figure \ref{PMAX} and Figure \ref{PMIN}, the number of patterns of the two compared algorithms always increases with an increase in \textit{maxSup}, and the number of patterns of the two compared algorithms always increases with a decrease in \textit{minSup}.

\subsection{Efficiency w.r.t. memory}

In this subsection, we further analyze the maximum memory occupied during the mining processes of the algorithms. According to the experimental results of six datasets, we found that in sparse datasets (e.g., retail, kosarak), the FRI-Miner algorithm can extract meaningful rare itemsets, that is, ignore most of the itemsets that do not meet the actual situation, but also the memory usage is lower than that of RP growth. However, for relatively dense datasets, a relatively large memory space is required during the mining process. In the dataset BMS, FRI-Miner shows good performance, not only mining more eligible itemsets but also taking up relatively less memory space. The experimental results of these datasets are shown in Figure \ref{databases Memory}(a) to Figure \ref{databases Memory}(f).
\section{Conclusion and Future Work} %
\label{sec:conclusion}

In this paper, we propose a novel fuzzy-theoretic-based algorithm for mining FRIs determined by fuzzy set theory and pattern mining. The purpose of this effective mining algorithm is to find meaningful rare itemsets that meet the minimum thresholds. In contrast to existing algorithms, an effective FRI-Miner algorithm is proposed based on a fuzzy list structure.  The algorithm requires specifying two minimum support thresholds; it utilizes several pruning strategies to prune unqualified itemsets and discovers the complete set of rare itemsets containing rare and frequent items. The experimental analysis shows that this algorithm performs well and has an improved overall mining quality  compared to that of the existing algorithm.

There are still some limitations in our research, such as the membership function and the specified minimum thresholds, which are defined in advance. If the threshold is specified to be extremely small, many candidate itemsets are generated, which occupy a large storage space. If it is specified as extremely high, some meaningful itemsets are ignored, which results in poor mining quality. It is hoped that there will be a more convenient algorithm in the future, which can automatically detect the most suitable threshold value and the setting of the membership function to effectively discover rare patterns.

\section*{Acknowledgment}

Thanks for the anonymous reviewers for their insightful comments, which  improved the quality of this paper. This work was partially supported by the National Natural Science Foundation of China (Grant No. 62002136 and Grant No. 11974373), and Guangzhou Basic and Applied Basic Research Foundation. 

\section*{References}
\bibliography{paper} 

\begin{thebibliography}{10}

\bibitem{agrawal1994fast}
Rakesh Agrawal, Ramakrishnan Srikant, et~al.
\newblock Fast algorithms for mining association rules.
\newblock In {\em 20th International Conference on Very Large Data Bases},
  pages 487--499, 1994.

\bibitem{berry2004data}
Michael~JA Berry and Gordon~S Linoff.
\newblock {\em Data mining techniques: for marketing, sales, and customer
  relationship management}.
\newblock John Wiley \& Sons, 2004.

\bibitem{chan1997mining}
Keith~CC Chan and Wai-Ho Au.
\newblock Mining fuzzy association rules.
\newblock In {\em Proceedings of the 6th International Conference on
  Information and Knowledge Management}, pages 209--215, 1997.

\bibitem{fournier2019mining}
Philippe Fournier-Viger, Yimin Zhang, Jerry Chun-Wei Lin, Hamido Fujita, and
  Yun~Sing Koh.
\newblock Mining local and peak high utility itemsets.
\newblock {\em Information Sciences}, 481:344--367, 2019.

\bibitem{gan2019survey}
Wensheng Gan, Chun-Wei Lin, Philippe Fournier-Viger, Han-Chieh Chao, Vincent
  Tseng, and Philip Yu.
\newblock A survey of utility-oriented pattern mining.
\newblock {\em IEEE Transactions on Knowledge and Data Engineering},
  33(4):1306--1327, 2021.

\bibitem{gan2017data}
Wensheng Gan, Jerry Chun-Wei Lin, Han-Chieh Chao, and Justin Zhan.
\newblock Data mining in distributed environment: a survey.
\newblock {\em Wiley Interdisciplinary Reviews: Data Mining and Knowledge
  Discovery}, 7(6):e1216, 2017.

\bibitem{gan2018survey}
Wensheng Gan, Jerry Chun-Wei Lin, Philippe Fournier-Viger, Han-Chieh Chao,
  Tzung-Pei Hong, and Hamido Fujita.
\newblock A survey of incremental high-utility itemset mining.
\newblock {\em Wiley Interdisciplinary Reviews: Data Mining and Knowledge
  Discovery}, 8(2):e1242, 2018.

\bibitem{han2011data}
Jiawei Han, Jian Pei, and Micheline Kamber.
\newblock {\em Data mining: concepts and techniques}.
\newblock Elsevier, 2011.

\bibitem{han2004mining}
Jiawei Han, Jian Pei, Yiwen Yin, and Runying Mao.
\newblock Mining frequent patterns without candidate generation: A
  frequent-pattern tree approach.
\newblock {\em Data Mining and Knowledge Discovery}, 8(1):53--87, 2004.

\bibitem{hemalatha2015minimal}
C~Sweetlin Hemalatha, Vijay Vaidehi, and R~Lakshmi.
\newblock Minimal infrequent pattern based approach for mining outliers in data
  streams.
\newblock {\em Expert Systems with Applications}, 42(4):1998--2012, 2015.

\bibitem{hong1999mining}
Tzung-Pei Hong, Chan-Sheng Kuo, and Sheng-Chai Chi.
\newblock Mining association rules from quantitative data.
\newblock {\em Intelligent Data Analysis}, 3(5):363--376, 1999.

\bibitem{hong1999data}
Tzung-Pei Hong, CS~Kuo, and SC~Chi.
\newblock A data mining algorithm for transaction data with quantitative
  values.
\newblock {\em Intelligent Data Analysis}, 3(5):363--376, 1999.

\bibitem{huang2012rare}
David Huang, Yun~Sing Koh, and Gillian Dobbie.
\newblock Rare pattern mining on data streams.
\newblock In {\em International Conference on Data Warehousing and Knowledge
  Discovery}, pages 303--314. Springer, 2012.

\bibitem{huang2014detecting}
David Tse~Jung Huang, Yun~Sing Koh, Gillian Dobbie, and Russel Pears.
\newblock Detecting changes in rare patterns from data streams.
\newblock In {\em Pacific-Asia Conference on Knowledge Discovery and Data
  Mining}, pages 437--448. Springer, 2014.

\bibitem{ji2012method}
Yanqing Ji, Hao Ying, John Tran, Peter Dews, Ayman Mansour, and R~Michael
  Massanari.
\newblock A method for mining infrequent causal associations and its
  application in finding adverse drug reaction signal pairs.
\newblock {\em IEEE Transactions on Knowledge and Data Engineering},
  25(4):721--733, 2012.

\bibitem{kim2007squire}
Chulyun Kim, Jong-Hwa Lim, Raymond~T Ng, and Kyuseok Shim.
\newblock {SQUIRE}: Sequential pattern mining with quantities.
\newblock {\em Journal of Systems and Software}, 80(10):1726--1745, 2007.

\bibitem{kiran2011novel}
R~Uday Kiran and P~Krishna Reddy.
\newblock Novel techniques to reduce search space in multiple minimum
  supports-based frequent pattern mining algorithms.
\newblock In {\em Proceedings of the 14th International Conference on Extending
  Database Technology}, pages 11--20, 2011.

\bibitem{koh2016unsupervised}
Yun~Sing Koh and Sri~Devi Ravana.
\newblock Unsupervised rare pattern mining: a survey.
\newblock {\em ACM Transactions on Knowledge Discovery from Data}, 10(4):1--29,
  2016.

\bibitem{koh2005finding}
Yun~Sing Koh and Nathan Rountree.
\newblock Finding sporadic rules using apriori-inverse.
\newblock In {\em Pacific-Asia Conference on Knowledge Discovery and Data
  Mining}, pages 97--106. Springer, 2005.

\bibitem{kuok1998mining}
Chan~Man Kuok, Ada Fu, and Man~Hon Wong.
\newblock Mining fuzzy association rules in databases.
\newblock {\em ACM SIGMOD Record}, 27(1):41--46, 1998.

\bibitem{le2018mining}
Tuong Le, Anh Nguyen, Bao Huynh, Bay Vo, and Witold Pedrycz.
\newblock Mining constrained inter-sequence patterns: a novel approach to cope
  with item constraints.
\newblock {\em Applied Intelligence}, 48(5):1327--1343, 2018.

\bibitem{lin2014mining}
Chun-Wei Lin and Tzung-Pei Hong.
\newblock Mining fuzzy frequent itemsets based on ubffp trees.
\newblock {\em Journal of Intelligent \& Fuzzy Systems}, 27(1):535--548, 2014.

\bibitem{lin2010efficient}
Chun-Wei Lin, Tzung-Pei Hong, and Wen-Hsiang Lu.
\newblock An efficient tree-based fuzzy data mining approach.
\newblock {\em International Journal of Fuzzy Systems}, 12(2):150--157, 2010.

\bibitem{lin2010linguistic}
Chun-Wei Lin, Tzung-Pei Hong, and Wen-Hsiang Lu.
\newblock Linguistic data mining with fuzzy {FP}-trees.
\newblock {\em Expert Systems with Applications}, 37(6):4560--4567, 2010.

\bibitem{lin2015rwfim}
Jerry Chun-Wei Lin, Wensheng Gan, Philippe Fournier-Viger, and Tzung-Pei Hong.
\newblock {RWFIM}: Recent weighted-frequent itemsets mining.
\newblock {\em Engineering Applications of Artificial Intelligence}, 45:18--32,
  2015.

\bibitem{lin2017fdhup}
Jerry Chun-Wei Lin, Wensheng Gan, Philippe Fournier-Viger, Tzung-Pei Hong, and
  Han-Chieh Chao.
\newblock {FDHUP}: Fast algorithm for mining discriminative high utility
  patterns.
\newblock {\em Knowledge and Information Systems}, 51(3):873--909, 2017.

\bibitem{lin2015fast}
Jerry Chun-Wei Lin, Ting Li, Philippe Fournier-Viger, and Tzung-Pei Hong.
\newblock A fast algorithm for mining fuzzy frequent itemsets.
\newblock {\em Journal of Intelligent \& Fuzzy Systems}, 29(6):2373--2379,
  2015.

\bibitem{linoff2011data}
Gordon~S Linoff and Michael~JA Berry.
\newblock {\em Data mining techniques: for marketing, sales, and customer
  relationship management}.
\newblock John Wiley \& Sons, 2011.

\bibitem{liu1999mining}
Bing Liu, Wynne Hsu, and Yiming Ma.
\newblock Mining association rules with multiple minimum supports.
\newblock In {\em Proceedings of the fifth ACM SIGKDD International Conference
  on Knowledge Discovery and Data Mining}, pages 337--341, 1999.

\bibitem{nguyen2019efficient}
Loan~TT Nguyen, Vinh~V Vu, Mi~TH Lam, Thuy~TM Duong, Ly~T Manh, Thuy~TT Nguyen,
  Bay Vo, and Hamido Fujita.
\newblock An efficient method for mining high utility closed itemsets.
\newblock {\em Information Sciences}, 495:78--99, 2019.

\bibitem{papadimitriou2005fuzzy}
Stergios Papadimitriou and Seferina Mavroudi.
\newblock The fuzzy frequent pattern tree.
\newblock In {\em The WSEAS International Conference on Computers}, pages 1--7,
  2005.

\bibitem{rymon1992search}
Ron Rymon.
\newblock Search through systematic set enumeration.
\newblock In {\em Proceedings of the Third International Conference on
  Principles of Knowledge Representation and Reasoning}, pages 539--550, 1992.

\bibitem{sadhasivam2011mining}
Kanimozhi~SC Sadhasivam and Tamilarasi Angamuthu.
\newblock Mining rare itemset with automated support thresholds.
\newblock {\em Journal of Computer Science}, 7(3):394, 2011.

\bibitem{shaw2001knowledge}
Michael~J Shaw, Chandrasekar Subramaniam, Gek~Woo Tan, and Michael~E Welge.
\newblock Knowledge management and data mining for marketing.
\newblock {\em Decision Support Systems}, 31(1):127--137, 2001.

\bibitem{srikant1996mining}
Ramakrishnan Srikant and Rakesh Agrawal.
\newblock Mining sequential patterns: Generalizations and performance
  improvements.
\newblock In {\em International Conference on Extending Database Technology},
  pages 1--17. Springer, 1996.

\bibitem{szathmary2007towards}
Laszlo Szathmary, Amedeo Napoli, and Petko Valtchev.
\newblock Towards rare itemset mining.
\newblock In {\em 19th IEEE International Conference on Tools with Artificial
  Intelligence}, pages 305--312. IEEE, 2007.

\bibitem{troiano2009fast}
Luigi Troiano, Giacomo Scibelli, and Cosimo Birtolo.
\newblock A fast algorithm for mining rare itemsets.
\newblock In {\em Ninth International Conference on Intelligent Systems Design
  and Applications}, pages 1149--1155. IEEE, 2009.

\bibitem{tsang2011rp}
Sidney Tsang, Yun~Sing Koh, and Gillian Dobbie.
\newblock {RP-T}ree: rare pattern tree mining.
\newblock In {\em International Conference on Data Warehousing and Knowledge
  Discovery}, pages 277--288. Springer, 2011.

\bibitem{zadeh1965fuzzy}
Lotfi~A Zadeh.
\newblock Fuzzy sets.
\newblock {\em Information and Control}, 8(3):338--353, 1965.

\end{thebibliography}

\end{document}